\documentclass[11pt]{article}
\usepackage[utf8]{inputenc}

\usepackage{amsmath,amssymb,amsthm,mathrsfs}
\usepackage{amsfonts,dsfont}
\usepackage{amsfonts}
\usepackage{gensymb}

\usepackage{natbib}
\usepackage{url}
\usepackage{amsfonts}
\usepackage{bbm}

\usepackage{mathtools}
\usepackage{bm}
\usepackage{graphicx}

\usepackage{authblk}

\usepackage{color}
\usepackage{caption}
\usepackage{subcaption}
\usepackage{float}
\newcommand{\ER}{Erd\H{o}s-R\'enyi}
\newtheorem{prop}{Proposition}

\usepackage{xcolor}
\definecolor{mycol}{rgb}{0,0,1}

\usepackage{setspace}
\begin{document}

\title{Credit Valuation Adjustment in Financial Networks}

\author[1]{Irena Barja\v{s}i\'{c}}
\author[2,3]{Stefano Battiston\footnote{Corresponding author. Email: stefano.battiston@uzh.ch}}
\author[4]{Vinko Zlati\'c}

\affil[1]{Department of Physics, Faculty of Science, University of Zagreb}
\affil[2]{Dept. of Banking and Finance, Univ. of Zurich, Switzerland} 
\affil[3]{Ca' Foscari Univ. of Venice, Italy}
\affil[4]{Theoretical Physics Division, Institute Rudjer Bov{s}kovi\'c, Bijeni\v{c}ka cesta 54, Zagreb, Croatia}

\date{Version: 
March 30th 2023}

\maketitle

\begin{abstract}
Credit Valuation Adjustment captures the difference in the value of derivative contracts when the counterparty default probability is taken into account. However, in the context of a network of contracts, the default probability of a direct counterparty can depend substantially on the default probabilities of indirect counterparties.
We develop a model to clarify when and how these network effects matter for CVA, in particular in the presence of correlation among counterparties defaults. We provide an approximate analytical solution for the default probabilities. This solution allows for identifying conditions on key parameters such as network degree, leverage and correlation, where network effects yield large differences in CVA (e.g. above 50\%), and thus relevant for practical applications. 
Moreover, we find evidence that network effects induce a multi-modal distribution of CVA values.
\end{abstract}

\section{Introduction}
Credit Valuation Adjustment (CVA) is the standard tool to quantify counterparty risk. In the presence of new information or changes in expectations about counterparties, CVA can give rise to substantial losses on portfolios even in the absence of actual counterparties' defaults\footnote{According to the Basel Committee on Banking Supervision \citep{basel2011capital,basel2015review}, about 2/3 of over-the-counter (OTC) credit derivatives losses in the 2008 financial crisis could be attributed to CVA}. 
Since 2010, CVA risk has become a core component of counterparty credit risk in the current international banking regulation framework of Basel III \citep{basel2019cva}. 
The importance of CVA has been increasingly recognised by scholars in a body of academic works investigating key aspects, including the pricing of portfolios of contracts
\citep{brigo2014arbitrage,bo2014bilateral}, 
wrong-way-risk \citep{brigo2018wwr,glasserman2018wwr}, and various important computational challenges \citep{chataigner2019credit,crepey2015bilateral}.

At the same time, there has also been an increased recognition that modelling the financial system as a network is a precondition to understand and manage financial stability from a macro-prudential perspective \citep{haldane2013rethinking,stiglitz2010risk}. This is widely reflected today in the policies of financial authorities both in the US \citep{yellen2013interconnectedness} and in the EU \citep{draghi2017building}, as network effects (both via counterparty risk and via common assets) also played had a key role in the 2008 financial crisis. Since then, building on early works \citep{eisenberg2001systemic}, the academic literature has developed an entire stream of new network models
\citep{stiglitz2010risk,acemoglu2015systemic,amini2016resilience,glasserman2016contagion,veraart2013failure,battiston2016price,bardoscia2017pathways,banerjee2022pricing}.

While both CVA and network effects have been investigated, alone, in various contexts, there is still a limited understanding of how CVA adjustment depends on the structure of the network of financial contracts, in particular in the presence of correlations across shocks on counterparties' assets. In this paper we develop a model to clarify precisely this issue.

In the case of a derivative contract between two parties, say A and B, CVA is computed by A as the risk-neutral expected value of the losses conditional to B's default. Practitioners typically compute such expected value under the assumption that the probability of default of B is governed by an exogenous stochastic process. While CVA captures adjustments in default probability of direct counterparties, it usually neglects adjustments in default probability of indirect counterparties, as pointed out in \citet{banerjee2022pricing}. However, if A and B are embedded in a network of contracts, then indirect counterparties of B can have a very important impact on B's default probabilities \citep{barucca2020network,banerjee2022pricing}.
Hereafter, we refer to Credit Valuation Network-Adjusted (CVNA) as the adjustment of credit contracts value computed taking into account the network effects, i.e. the default probabilities of all agents in the system, as opposed to the adjustment solely based on the default probability of direct counterparties, which we denote as CVDA (Credit Valuation Direct-Adjusted) when it is otherwise ambiguous. 
In this paper, we investigate 1) under which conditions, CVNA can substantially larger in magnitude compared to CVDA and 2) how the difference between CVNA and CVDA depends both on the network structure and on the correlation among shocks hitting the assets of all counterparties in the system. 

In line with the literature on financial networks \citep[see e.g.,][]{eisenberg2001systemic,amini2016resilience,veraart2013failure,barucca2020network}, we model a network of financial agents, each holding a portfolio of other agents' liabilities as well as a portfolio of \textit{external} assets, i.e. that are shocked exogenously. However, we interpret the linkages more generically as derivative contracts.
We are aware that if counterparties establish derivative contracts having as underlying other counterparties in the system, the existence and uniqueness of the solution to the classic clearing problem stated in \citet{eisenberg2001systemic} breaks down as shown in \citet{schuldenzucker2020default}. Here, we avoid this case by assuming that all contracts have as underlying entities that are exogenous to the system. 

A first innovation compared to the above mentioned literature, is our specific focus on the difference between CVNA and CVDA and its drivers. To our knowledge this has not been quantified systematically. 
In our model, CVNA is proportional to the average of the default probability $q$ of a counterparty, computed taking into account the whole network, while CVDA is proportional to the average default probability, denoted as $p_1$, of a counterparty when only its intrinsic sources of shock are considered.  
A second, key innovation in our model regards the introduction of correlation among agents' portfolios of external assets. To keep tractability both in terms of analytical and numerical computations, but without loss of generality, we capture such correlation by means of the correlation among shocks hitting the agents' equity levels.



In such setting, we are able to derive an approximated analytical solution for the agents' default probability $q$ as a function of key parameters such as network degree $k$, agents' financial leverage $\Lambda$ and correlation $\rho$ among agents' shocks. 
We identify a wide region in the parameter space where $CVNA >> CVDA$. In simplified terms, the conditions for that to hold can be characterised as follows: 
i) in absence of correlation,either low enough leverage and low enough network degree, or, conversely, high enough leverage, across all levels of network degree;
ii) in the presence of correlation: high enough leverage, in this case for all degree levels.
In other words, with high enough leverage, we find $CVNA >> CVDA$ both with and without correlation.

Moreover, we find that correlation can lead to the fact that the default probability distribution 
becomes multimodal with non negligibile probability mass close to 1. Therefore care is required when the average value of the default probability is used for CVA. 


\subsection{Related work}
This paper is related to several strands of literature. 

First, a recent body of work focuses on improving the standard CVA computation. \citet{bo2014bilateral} have obtained an analytical framework for calculating the bilateral CVA for a large portfolio of credit default swaps. Wrong way risk is addressed both in works of \citet{brigo2018wwr} and \citet{glasserman2018wwr}, where in the first paper a semi-analytical approach is proposed instead of the usual numerical techniques, and in the second one the bounds for the CVA are obtained by using marginal distributions of credit and market risk, and varying the dependence between them. The use of genetic algorithms for optimizing the portfolio CVA is explored in \citet{chataigner2019credit}. \citet{crepey2015bilateral} develop a way to calculate CVA under funding constraints using reduced-form backward stochastic differential equations. \citet{abbas2022pathwise} use neural networks to reduce the computation time for a pathwise CVA calculation.

Since in our model we compute the fixed point of a clearing problem, conditioned to external shock, our paper is related to the literature considering the problem of financial clearing of the entire network of contracts. The question of how much each agent owes to whom was formalised as a clearing problem in the seminal paper of \citet{eisenberg2001systemic}, one of the first work to model the financial dependencies as a directed network. The problem of the interdependence of payments is turned into a fixed point problem, and the authors show that a unique clearing vector exists under very general conditions on the network structure, for contract of type ``long''. A more realistic setting with non-zero bankruptcy costs is considered in \citet{veraart2013failure}. \citet{suzuki2002valuing} deal with the effect of cross-holdings of equities, while \citet{bernstein2018dynamic} introduce the time dynamics of the interbank liabilities in the Eisenberg and Noe model and then study the default risk dynamically. \citet{schuldenzucker2020default} show that, if financial contracts are of type ``short'', such as credit default swaps, then the solution to the clearing problem might not be unique, or even cease to exist.

Our work is also related to the strand of literature concerning the network valuation frameworks, i.e. approaches to compute the valuation of contracts, whose values are endogenously interdependent via the network.
In \cite{gourieroux2012bilateral}, a decomposition of an exogenous shock into the direct and contagion part is provided, however, the exact solution to the system of equations is limited by the number of banks it can be applied to. 
\citet{banerjee2022pricing} employ a comonotonicity framework to the network of banks to circumvent the ``curse of dimensionality'' in the analytical treatment of the valuation problem.
\citet{barucca2020network} provide a general network valuation framework that encompasses both clearing models and ex-ante valuation models and accounts for the uncertainty regarding the external assets. 

For our numerical computations, we build on the results of \citet{gai2010contagion,KOBAYASHI2014model}, who showed that some classes of financial contagion models have an analytical representation in the form of the threshold model by \citet{watts2002threshold}. We rely in this insight to provide an analytical framework that includes correlated shocks. 
\citet{KOBAYASHI2014model} consider external asset returns that are normally distributed but without correlation. \citet{kobayashi2013network} consider a network of 5 banks and introduce correlations on their external assets, for some specific examples of the dependence structure. In \citet{nowak2011individual} a network of 5 banks is used as well to study the systemic cost of bank failures, considering also the degree of correlations on external asset prices as a parameter to vary. 

Some analytical results on the expected size of the default cascades for large networks were obtained, e.g., by \cite{battiston2012default,roukny2013default,amini2016resilience}.
The analysis of cascades in the general framework of threshold models was mainly covered in the physics literature \citep{gleeson2007seed,gleeson2013binary}, where, differently from our case, the motivation for introducing the correlations on thresholds was missing. 

\section{Model}

\subsection{Model description}

We investigate a system of $N$ banks (referred to also as ``vertices'' depending on the context) interconnected in a network of credit contracts. 
%
In addition to contracts among banks, banks have holdings of external assets, e.g. securities issued by actors outside the set of banks. 
At time $t=0$, all the investments are made, and for each bank $i$ they are represented as a balance sheet consisting of interbank assets $A^b_i = \sum_j A^b_{ij}$, interbank liabilities $L^b_i = \sum_j L^b_{ij} = \sum_j {A^b_{ij}}^T$, external assets $A^e_i = \sum_j A^e_{ij}$, and external liabilities $L^e_i$. 
In general, the difference between the assets and liabilities, the equity $E_i(t)$, is defined at time $t$ as:
\begin{align}
    E_i(t) &= A^e_{i}(t) + A^b_{i}(t) - L^e_i(t) - L^b_{i}(t) \nonumber \\
    &= {\Lambda}^e \sum_k E_{ik} x^e_k(t) + {\Lambda}^b \sum_j B_{ij} x^b_j(t) - L^e_i - L^b_i
\label{eqn:balance_sheet}
\end{align}
The decomposition of external and interbank assets bears resemblance. Both the external and interbank leverage, ${\Lambda}^e_i = A^e_i/E_i$ and ${\Lambda}^b_i = A^b_i/E_i$ are set to be the same for every bank: ${\Lambda}^e_i = {\Lambda}^e$,  ${\Lambda}^b_i = {\Lambda}^b$. The matrices $E_{ik}$ and $B_{ij}$ represent  the structure of the external investments of the banks and  
the adjacency matrix of the underlying interbank network, respectively. The elements of $E_{ik}$ and $B_{ij}$ are normalized as fractions of the respective total assets. Unitary values of investments $x^e_k(t)$ and $x^b_j(t)$ are initially equal to $1$. 


First, we explain the time dependence of the unitary values of interbank investments $x^b_j(t)$. We define the indicator variable ${\chi}_i(t)$, that indicates the state of default from the equity level: ${\chi}_j(t) = \mathbbm{1}_{E_j(t) < 0}$. The values of interbank investments depend on the counterparty's default state, 
and can, accordingly, take two possible values $x^b_i({\chi}_i = 0) = 1$ and  $x^b_i({\chi}_i = 1) = \delta \cdot R_i$. In the case the counterparty has not defaulted yet, the indicator variable stays equal to $1$ as initially set, and in the case of default, it gets reduced to $\delta \cdot R_i \in [0,1\rangle$. The recovery rate is represented as the product of an exogenous recovery rate $\delta \in [0,1\rangle$ and an endogenous recovery rate $R_i$, that is defined as:
\begin{equation}
    R_i = \frac{\min \Big\{\max \Big\{0, A^e_{i}(t) + A^b_{i}(t) - L^e_i(t)\Big\}, L_i^b(t) \Big\}}{L^b_{i}(t)}
\end{equation}
$R_i$ equals the fraction of the interbank liabilities the defaulted bank can still repay with its remaining assets, after having repaid its external liabilities. The exogenous recovery rate $\delta$ represents frictions in the process of defaults such as legal costs. In this paper we focus on the worst case scenario, when the exogenous recovery rate $\delta$ is equal to 0.

For the purpose of this paper, we disregard the structure of the external asset investments, and represent them simply as $A^e_i(t)$. We have then as the balance sheet equation:
\begin{align}
    E_i(t) &= A^e_i(t) + {\Lambda}^b \sum_j B_{ij} x^b_j(t) - L^b_i - L^e_i
\end{align}
Since we set the total interbank assets and liabilities to be the same for every bank $i$, they also have to be equal to each other in amount, due to the relationship $L^b_{ij} = {A^b_{ij}}^T$. 

For the purpose of this analysis we introduce some homogeneity among the financial agents that makes the treatment more tractable. Heterogeneity could be introduced at a later stage. 
We represent the shock propagation process as occurring in successive steps denoted by $t= 0, \, 1, ...$. We set the initial value of external assets $A^e_i(t=0)$ equal to one, and the external liabilities $L^e_i(t=0)$ equal to zero. This implies that the initial total equity is equal to one, $E_i(t=0) = 1$, for each bank $i$. After the process of shock propagation has started, a bank $i$ is considered to default when the equity value goes below zero, $E_i < 0$. 

At time step $t=1$, the external assets of bank $i$ incur a shock $s_i$ and take the value $A^e_i(0)(1+s_i)$.
The computation of the equity and the determination of defaults proceeds as follows. 
After the shock is first applied to the bank's external assets, the value of total assets is updated (as the sum of interbank assets and the remaining external assets). Then the value of banks' total assets is compared to banks' total liabilities to check whether any bank has defaulted. 
In case of default of a bank $k$, the value of its remaining assets is calculated and allocated to each of $k$'s counterparties proportionally to its interbank liabilities to each counterparty. Moreover, the value is also multiplied by the exogenous rate of recovery, $\delta$. The new value of the interbank assets is thus equal to ${\Lambda}^b \sum_j B_{kj} \delta \cdot R_j$, where the counterparties of a bank $k$ are denoted by $j$.
Because banks' interbank assets have changed some new default may occur. The process of adjusting the value of interbank assets and checking for default continues until no new defaults occur. 

The process of the propagation of the shock through the network can be written in terms of the following map of the vector of levels of equity on itself, component-wise: 
\begin{equation}
    E_i(t+1) = A^e_i(0)(1+s_i) + {\Lambda}^b \sum_j B_{ij} x^b_j(\mathbbm{1}_{E_j(t) < 0}) - L^b_i, \quad x^b_j(0)=1, \quad x^b_j(1)=\delta \cdot R_i.
\end{equation}
It is possible to show that the map is non decreasing and hence the fixed point denoted with 
$E_i^* = E_i(T)$ is unique \citep{barucca2020network}.
The propagation is equivalent to a process presented in \cite{veraart2013failure}, with the constraint $\alpha = \beta$, and $\beta$ being equivalent to our $\delta$.
According to the notation in \cite{eisenberg2001systemic}, we can define the financial system as a triplet $(\mathbf{A^b}$, $L^b$, $\overline{A^e}$) with the addition of the exogenous recovery rate $\delta$. $\mathbf{A^b}$ is a matrix representing the network of interbank assets, $L^b$ is a vector of interbank liabilities of each bank, and $\overline{A^e}$ is a vector defined as $\overline{A^e} = A^e - L^e$, which represents the net amount of assets of each bank that are external to the interbank network.
As we will show in the paragraph below, shock vectors are modelled as a random variable, hence the equilibrium number of banks in defaults is also a dependent random variable. We are interested in the distribution of equilibrium states. 

\subsection{Modelling shocks and their correlations}
For a set of N banks, we sample a shock vector $\vec{s}$ of dimension $N$, where each component $s_i$ represents the shock on bank $i$.
Our aim here is to retain tractability while capturing in the simplest way the fact that banks experience shocks of different magnitude. Accordingly, shocks take on values out of a discrete set : $s_i = \{ {\sigma}_{\mu} \}$ with $\mu=1,...,n_{\sigma}$, each with probability $p_{\mu}$. 
The values ${\sigma}_{\mu}$ are set as follows. 
The value of ${\sigma}_{1}$ is large enough to cause the immediate default of the affected bank, i.e. ${\sigma}_{1} < -1$. The values shocks ${\sigma}_{\mu}, \mu=2,...,n_{\sigma}$ correspond to levels of losses on the banks' equity that do not cause immediate default but weaken the bank financial situation, i.e. ${\sigma}_{\mu} \in [-1,0 \rangle$. The last element is set to zero, ${\sigma}_{n_{\sigma}} = 0$, and it represents the case of no loss on equity. In this paper, we focus on the simplest case, $n_{\sigma} = 3$, with three different types of shocks: immediate default, partial loss on equity, and no loss. We group the affected vertices into compartment $\mu$, depending on the shock ${\sigma}_{\mu}$ they received, and denote the number of vertices in compartment $\mu$ as $N_{\mu}$.

We define the random variable of the shock value as  $\mathcal{S}$, which realizes shock values $\sigma_\mu$ with the previously defined respective probabilities $p_{\mu}$. The cumulative distribution function of random variable $\mathcal{S}$ will have the form:
\begin{equation}
    F_\mathcal{S}(\mathcal{S}=\sigma_x) = \sum_{\mu=0}^{x} p_{\mu}, \quad x = {\sigma}_1,{\sigma}_2,\dotsc,{\sigma}_{n_{\sigma}}
    \label{eqn:shock_sample}
\end{equation}

The sampling of values of the shock random variable $S$ is then performed using the inverse of this CDF on uniform random variables, $F_\mathcal{S}^{-1}(U)$.

A number of examples from the financial systems point to the existence of correlation between the external shocks. Companies that belong to the same industry, or the same geographic region will be affected by similar exogenous events. In some time periods, adverse general economic conditions can cause a higher correlation of defaults \citep{sandoval2012correlation,hull2012options}. Therefore, to obtain a more realistic picture of the shock structure, we need to introduce correlations on them by correlating the uniform random variables. For that purpose, we use the copula \citep{sklar1959fonctions,schmidt2007copula}, a function that isolates the dependency structure of a multivariate distribution. In particular, out of the several families of copulas, we choose the Gaussian copula for this analysis, since it is analytically most tractable. In addition to that, for a single factor Gaussian copula \citep{krupskii2013factor}, i.e. a copula that has all the correlations set to the same value, the Gaussian random variable can be decomposed into a common and an idiosyncratic random variable, that are mutually independent, which is a property we will use in the Appendix C to derive the expression we need.
The Gaussian copula is represented as a cumulative distribution function of the multivariate Gaussian distribution $F_{\textbf{Z}}$, with marginal distributions transformed to the uniform distribution using the univariate Gaussian cumulative distribution functions $F_{Z_i}$:
\begin{equation}
    C_{\Sigma}^{Ga}(u_1,u_2,\dotsc,u_N) = F_{\textbf{Z}}(F_{Z_1}^{-1}(u_1),F_{Z_2}^{-1}(u_2),\dotsc,F_{Z_N}^{-1}(u_N))
\end{equation}

The sampling of the uniform variables all correlated with $\rho$ then goes in the following way: we take an n-variate Gaussian probability density function $f_{\textbf{Z}}$, with the covariance matrix $\Sigma$ that has diagonal terms equal to $1$ and all off diagonal terms equal to the same chosen level of correlation, ${\rho}_{ij} = \rho$.
We sample an n-dimensional vector $\vec{z} = (z_1,\dotsc,z_N)$ from the multivariate Gaussian distribution, and transform it into an n-dimensional sample $\vec{u} = (u_1,\dotsc,u_N)$ from the uniform distribution, using the CDF of the Gaussian univariate distribution. 

The sample of uniform random variables that we obtained is now correlated, and we can proceed in transforming it to a sample of shock values $\vec{s} = (s_1,\dotsc,s_N)$, using the inverse of a discrete CDF $F_S^{-1}(U)$.

\subsection{Network topology}
In order to model the propagation of the shocks we defined in the previous section, we employ the random k-regular graph. This topology introduces randomness in the pairs of vertices, while keeping the vertex degree fixed. Under our constraints for the equity and the interbank leverage, randomness in the vertex degree would greatly extend the time needed for the simulations.
In addition to reducing the length of the simulations, a simple network topology enables us to have more control over the parameters. 

Equal interbank leverage is assigned to each bank ${\Lambda}_b$ and divided over its interbank assets. The condition that the equities and interbank leverage are uniform over all the banks ($E_i = E$, ${\Lambda}^b_i = {\Lambda}^b$) translates into all the assets and liabilities having the same value, $A^b_i = A^b$, $L^b_i = L^b$. Since the assets and liabilities can be represented as a network, from $L^b_{ij} = A^b_{ij}$ we recover $A^b = L^b$. 

On a random k-regular network, where all the vertices have the same degree $k$, these constraints lead to a solution where all the individual assets and liabilities are equal in value, $A^b_{ij} = L^b_{ij} = \frac{A^b}{k/2}$. In comparison, using an \ER{} graph would require iteratively looking for a solution for $A^b_{ij}$ and $L^b_{ij}$, which is not guaranteed to exist for every instance of an \ER{} graph. 

The triplet representing the financial system  ($\mathbf{A}^b$, $L^b$, $\overline{A^e}$) in this case consists of a matrix with values $\frac{{\Lambda}^b}{k/2}$ everywhere and zeros on the diagonal, a vector composed of all values equal to ${\Lambda}^b$ and a vector with all values equal to 1.

\subsection{Process simulation}
In the first step, the external shock is applied to the external assets. If it triggers a bank to default, the default is further propagated with the clearing algorithm described in Model description, until the condition, that all the equity changes are less than $3 \%$ of the original equity, is fulfilled. In case the condition cannot be fulfilled, there is a cutoff set after $c$ steps. The exogenous recovery rate is set to $\delta = 0$.

Shocks are sampled 1000 times and imposed on the external assets of the vertices in a network. For every network degree $k$, 5 different instances of the random network are generated. On every network instance, 1000 process realizations $\xi$ are started using the sample of shocks, resulting in a total of $N_{\xi} = 5000$ process realizations. For each realization a vector of default indicators for banks is obtained. The fraction of defaulted vertices within the realization $\xi$, $q(\xi)$, is calculated as the expected value of the default indicator vector.
We then obtain a distribution of default fractions from entire sample with the same network degree $k$, and use its expected value $\langle q \rangle = \frac{1}{N_{\xi}} \sum_{\xi} q(\xi)$ as an estimation for the probability of default of a counterparty $q$, computed taking into account the whole network.

\subsection{Threshold model}
It has been shown in \cite{gai2010contagion,KOBAYASHI2014model}, that the financial contagion model, based on balance sheets and with variable external assets, is equivalent to the threshold model of \cite{watts2002threshold}. 
We use the mean field approximation to obtain the solution for the threshold model. For both the correlated and uncorrelated case, we obtain the distribution of the banks shocked by each value of shock ${\sigma}_{\mu}$, and use it to calculate the expected fraction of defaults $\langle q \rangle$, in addition to obtaining it from the simulations. The analytical approach enables us to study the behaviour of the expected fraction of defaults in the infinite system size limit.

\subsection{CVNA vs. CVDA}
We denote the standard CVA measure as CVDA (Credit Valuation Direct-Adjusted) to stress that the valuation adjustment is made taking into account only the probability of default of the direct counterparty. Therefore, we can state that it is proportional to the exogenous probability of default of a counterparty $p_1$, $CVDA \propto p_1$.
On the other hand, we introduce the CVNA (Credit Valuation Network-Adjusted) to introduce the network effects to the probability of default of the counterparty into the valuation of the contract. CVNA is then proportional to the above defined probability of default of a counterparty that takes
into account the whole network, $\langle q \rangle$, $CVDA \propto \langle q \rangle$.
If we consider all else that contributes to both adjustments to be the same, we can write the following relation:
\begin{equation}
    CVNA = CVDA \cdot \frac{\langle q \rangle}{p_1}
\label{eqn:CVNA-CVDA}
\end{equation}
which states that CVNA will significantly differ from CVDA in cases when $\langle q \rangle$ is a result of a significant amplification of $p_1$.


\section{Results}

\subsection{Threshold model - a mean field approximation}

The mapping between the financial contagion and the threshold model of \citet{watts2002threshold} has been shown by \cite{KOBAYASHI2014model}. We will shortly explain the reasoning behind this. 
The threshold model is set up as follows; every agent in the system that is at some state 0, observes the binary states of its $k$ neighbouring agents, and then changes its states according to a simple rule; if a threshold fraction $\varphi$ of the total of $k$ neighbours is in state 1, it switches to state 1, else it remains in state 0. In the case of financial contagion, the two states of interest are default and non-default, and the thresholds can be easily calculated considering the following.

Since our default condition is $E_i < 0$, bank $i$ defaults either by being hit with the initial default shock $p_1$, or by having a large enough number of neighbours that default.
For a central vertex with the level of equity after the shock equal to ${\epsilon}_{\mu}$, the required number of neighbours $m_{\mu}$, that need to default to cause its default, is:
\begin{equation}
    m_{\mu} = \left \lceil \frac{k}{2} \frac{{\epsilon}_{\mu}}{{\Lambda}_b} \right \rceil \leq \frac{k}{2}
\end{equation}
The number of neighbours $m_{\mu}$ then represents the threshold from the threshold model described above. 

In our setup, the external shocks, sampled from $\vec{\sigma} = ({\sigma}_1,{\sigma}_2,{\sigma}_3)$ and delivered to the system, are then simply reflected in the initial distribution of different thresholds of equity values $\vec{\epsilon} = ({\epsilon}_1,{\epsilon}_2,{\epsilon}_3)$. We group the vertices depending on the type of the shock they received into compartments $\vec{N} = (N_1,N_2,N_3)$.
The underlying network is assumed to be a random regular network, in which every node has the same degree $k$, which is the sum of in- and out-degrees $k = k_{in} + k_{out}, k_{in}=k_{out}$.

\begin{prop}
The equilibrium value of the default probability in the financial system defined earlier is given by:
\begin{equation}
    \langle q \rangle = \sum_{N_1,N_2,N_3} \binom{N}{N_1,N_2,N_3}\quad p_1^{N_1} \cdot p_2^{N_2} \cdot p_3^{N_3} \cdot q(\vec{N})
\label{eqn:q_nocorr}
\end{equation}
for the case with no correlation on shocks. When correlations exist on the shocks \citep{krupskii2013factor}, in the manner defined earlier, the equilibrium value follows as:
\begin{equation}
    \langle q \rangle = \int_{-\infty}^{\infty} f_{X(\rho)}(\alpha) \left( \sum_{N_1,N_2,N_3} \binom{N}{N_0, N_1, N_2}  {\pi}_1(\alpha)^{N_1} {\pi}_2(\alpha)^{N_2} {\pi}_3(\alpha)^{N_3} \right) \times q(\vec{N})\,d\alpha  
\label{eqn:q_corr}
\end{equation}
with ${\pi}_1(\alpha), {\pi}_2(\alpha), {\pi}_3(\alpha)$ as:
\begin{align}
 & {\pi}_1(\alpha) \coloneqq F_{Y}(z_I - \alpha) = F_{Y}(F_Z^{-1}(p_I) - \alpha) \\ \nonumber
    & {\pi}_2(\alpha) \coloneqq F_{Y}(z_{II} - \alpha) - F_{Y}(z_I - \alpha) = F_{Y}(F_Z^{-1}(p_{II}) - \alpha) - F_{Y}(F_Z^{-1}(p_I) - \alpha) \\ \nonumber
    & {\pi}_3(\alpha) \coloneqq 1 - F_{Y}(z_{II} - \alpha) = 1 - F_{Y}(F_Z^{-1}(p_{II}) - \alpha) \\ 
\end{align}
and $p_I \coloneqq p_1$ and $p_{II} \coloneqq p_1 + p_2$, $z_I \coloneqq F_Z^{-1}(p_I)$, $z_{II} \coloneqq F_Z^{-1}(p_{II})$, where $F_Z^{-1}$ and $F_Y^{-1}$ are inverse cumulative distribution functions of random variables $Z$ and $Y$. $Z_i$ is a normal random variable $Z_i \sim \mathcal{N}(0,1)$ and $Y_i$ comes from the decomposition $Z_i = X + Y_i, \quad Z_i \sim \mathcal{N}(0,1), \quad  X \sim \mathcal{N}(0,\rho), \quad Y_i \sim \mathcal{N}(0,1-\rho)$.
\end{prop}

\begin{proof}
We can write an iterative equation for the expected fraction of default $q_n$ using the mean field approximation, which assumes the homogeneity of the system, and the lack of all correlations between the variables. Since we consider the system to be homogeneous, we track the average default state of a central vertex, $q_n$. The initial value of the expected fraction of default $q_0$ is equal to the exogenous probability of default $p_1$. Assuming the lack of correlations between the neighbours allows us to use the binomial distribution to calculate the probability that at least $m_{\mu}$ neighbours defaulted for a central vertex that was shocked with ${\sigma}_{\mu}$.  Summing over all possibilities for shocks $\mu$, we get the iterative expression: 
\begin{equation}
 q_n = \sum_{\mu} \frac{N_{\mu}}{N} \sum_{m=m_{\mu}}^{k/2} \binom{k/2}{m} q_{n-1}^{m} (1 - q_{n-1})^{k/2-m} = f(q_{n-1})
\label{eqn:iter}
\end{equation}
The solution to this equation can be obtained by looking for a fixed point for $q^* = f(q^*) = \lim_{n \to \infty}f(q_n)$.

For any possible realization of the vector of vertices $\vec{N}$ resulting from some shock, if we know the degree of the network $k$, the interbank leverage ${\Lambda}_b$ and the equity sizes after the shock, $\vec{\epsilon} = \vec{1} - \vec{\sigma}$, the fixed point solution provides us the expected fraction of defaulted vertices:
\begin{equation}
    q(\Vec{N}) = \sum_{\mu} {\Pi}_{\mu}(N_{\mu}) \left[1 - F_{CDF}^{BINOM}\left(m_{\mu},\frac{k}{2},q(\vec{N}) \right)\right], \quad {\Pi}_{\mu} = \frac{N_{\mu}}{N}, \quad m_{\mu} = \left \lceil{\frac{k}{2} \frac{{\epsilon}_{\mu}}{{\Lambda}_b}}\right \rceil \leq \frac{k}{2}
    \label{eqn:q_N}
\end{equation}
To calculate the expected fraction of defaulted vertices $\langle q \rangle$ over all the possible shock realizations, i.e. represented in this calculation with all the possible values of the vector $\vec{N}$, we need the probability distribution $\mathcal{P}(\vec{N};N;\vec{p})$ of $\vec{N}$, given the individual initial shock probabilities $\vec{p} = (p_1,p_2,p_3)$:
\begin{equation}
    \langle q \rangle = \sum_{N_1,N_2,N_3} \mathcal{P}(\vec{N};N;\vec{p}) \cdot q(\vec{N})
\end{equation}
The probability distribution in question for the non-correlated case is simply the multinomial distribution, so the final expression looks like Eqn \ref{eqn:q_nocorr}.
In the Appendix C we provide the derivation of the probability distribution for the correlated case (Eqn \ref{eqn:p_corr}), which we plug into the previous expression to get Eqn \ref{eqn:q_corr}.
\end{proof}
\medskip

To shorten the time necessary for the numerical evaluation of Eq. \ref{eqn:q_corr}, we approximate the varying multinomial distribution by using only the expected values, $\mathbb{E} [N_{\mu}(\alpha)] = N {\pi}_{\mu}(\alpha)$ which then replace the factor ${\Pi}_{\mu} = N_{\mu}/N$ with $\mathbb{E} [{\Pi}_{\mu}] = \mathbb{E}[N_{\mu}(\alpha)]/N = {\pi}_{\mu}(\alpha)$ in the Equation \ref{eqn:q_N}:
\begin{equation}
    q(\alpha) = \sum_{\mu} {\pi}_{\mu}(\alpha) \left[1 - F_{CDF}^{BINOM}\left(m_{\mu},\frac{k}{2},q(\alpha) \right)\right], m_{\mu} = \left \lceil{\frac{k}{2} \frac{{\epsilon}_{\mu}}{{\Lambda}_b}}\right \rceil \leq \frac{k}{2}
\label{eqn:q_alpha}
\end{equation}
The expected fraction of defaults then becomes:
\begin{equation}
    \langle q \rangle = \int_{-\infty}^{\infty} f_X(\alpha) \cdot q(\alpha) \,d\alpha
\label{eqn:q_corr_s}
\end{equation}

We show the comparison of the Eqn. \ref{eqn:q_nocorr} with the simulation results in Fig. \ref{fig:nocorr}. The comparison of the result of Eqn. \ref{eqn:q_corr_s} and the simulation results is depicted in Fig. \ref{fig:corr}. We can see from the Figure \ref{fig:theory} that, in case when there are no correlations, CVNA ($\sim \langle q \rangle$) has a higher value compared to CVDA ($\sim p_1$), for intermediate leverages in the lower range of the network degrees $k$, or high enough leverages for all $k$. In case with correlations, CVNA is larger than CVDA for any network degree $k$ if the leverage is high enough.

\begin{figure*}[!htbp]
\begin{subfigure}{.5\textwidth}
  \centering
  \includegraphics[width=1\linewidth]{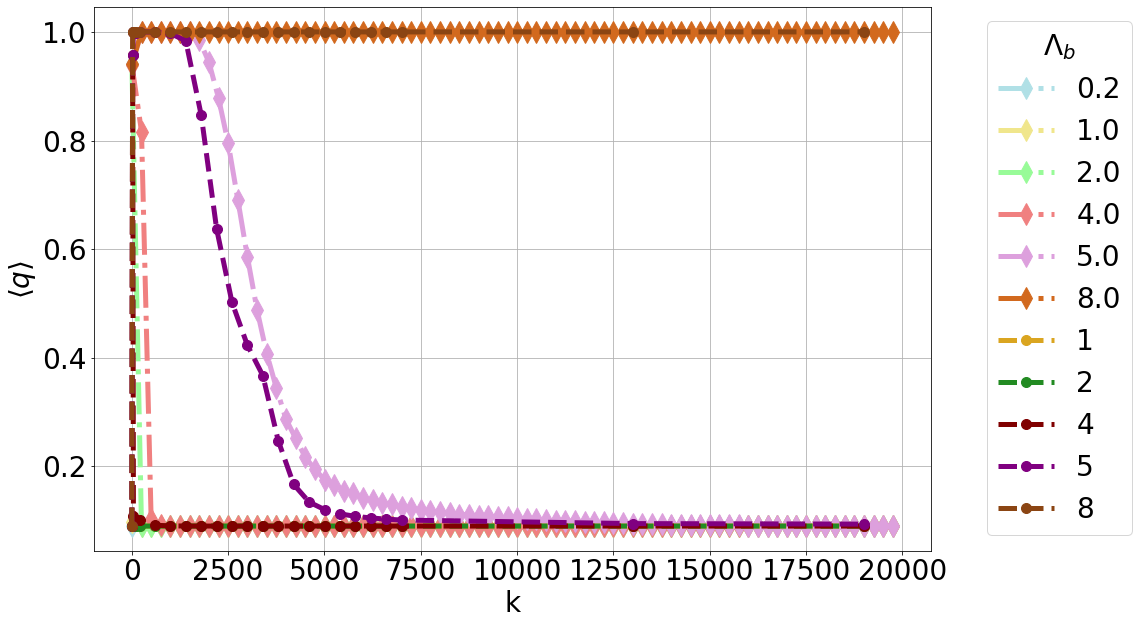} 
  \caption{\textbf{$N = 10000$,$\rho = 0$}}
  \label{fig:nocorr}
\end{subfigure}%
\begin{subfigure}{.5\textwidth}
  \centering
  \includegraphics[width=1\linewidth]{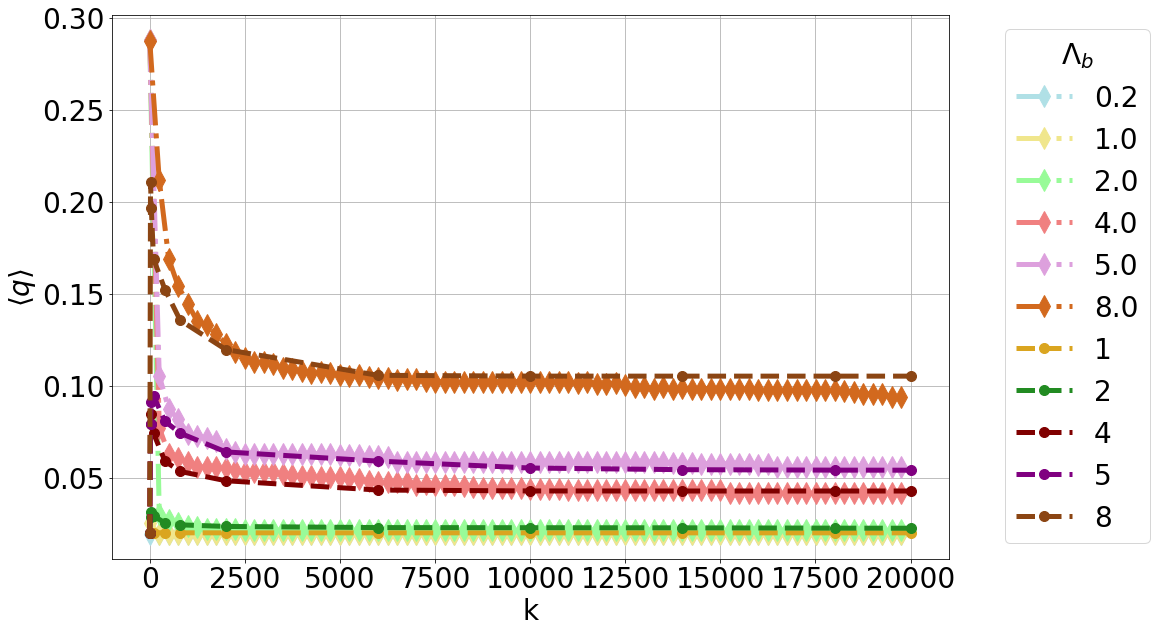} 
  \caption{\textbf{$N = 10000$, $\rho = 0.3$}}
  \label{fig:corr}
\end{subfigure}
\caption{\textbf{Comparison of the theoretical calculation and the simulation results.} Both panels show results on the network of size $N = 10000$. Panel \textbf{a)} shows the probabilities of default $\langle q \rangle$ depending on the network degree $k$ for the shock with parameters $\vec{p} = (0.09,0.083,0.827)$, $\vec{\sigma} = (-1.1,-0.5,0.0)$ and no correlation between shocks, $\rho = 0$.
Panel \textbf{b)} shows the probabilities of default $\langle q \rangle$ depending on the network degree $k$ for the shock with parameters $\vec{p} = (0.02,0.09,0.89)$, $\vec{\sigma} = (-1.1,-0.5,0.0)$ and with correlation $\rho = 0.3$ between shocks. }
\label{fig:theory}
\end{figure*}

Prop. \ref{prop:q_infinite_network}
provides the expression for the average default probability in the case of an infinite network, in presence of correlation.

\begin{prop}\label{prop:q_infinite_network}
In the case of an infinite network, the limiting equilibrium value of the default probability for the case without correlation on shocks is given by:
\begin{equation}
    \langle q_{lim} \rangle =  p_1 + \sum_{\mu=2} p_{\mu} \cdot \prod_{\phi = 2}^{\mu} \Theta\left(\sum_{\nu=1}^{\phi-1} p_{\nu} - \frac{{\epsilon}_{\phi}}{{\Lambda}_b}\right)
    \label{eqn:lim_nocorr}
\end{equation}
When correlations are present, the limiting equilibrium value follows as:
\begin{align}
    \langle q_{lim} \rangle &= \int_{-\infty}^{\infty} f_X(\alpha) {\pi}_1(\alpha)\,d\alpha + \int_{-\infty}^{l_{II}} f_X(\alpha) {\pi}_2(\alpha)\,d\alpha + \int_{-\infty}^{\min{(l_{II}, l_{III} )}} f_X(\alpha) {\pi}_3(\alpha)\,d\alpha, \nonumber \\ 
    \quad l_{II} &= F_Z^{-1}(p_I) - F_Y^{-1}\left( \frac{{\epsilon}_2}{{\Lambda}_b} \right),
    \quad l_{III} = F_Z^{-1}(p_{II}) - F_Y^{-1}\left( \frac{{\epsilon}_3}{{\Lambda}_b} \right)
    \label{eqn:lim_corr}
\end{align}
The minimum value of $\langle q_{lim} \rangle$ is still $p_1$, like in the uncorrelated case, and it comes from the first integral in Equation \ref{eqn:lim_corr}. The value of other two integrals is always positive and defined by the limits $l_{II}$ and $l_{III}$.
\end{prop}
\begin{proof}
We study the behaviour of the default probability in the limit of an infinitely large network ($N \to \infty$) that is complete ($k \to \infty$). 
To start, we consider the Equation \ref{eqn:q_N} in that light.
We wish to obtain the form of the cumulative distribution function of the binomial distribution in the limit of $k \to \infty$. We observe that the standard deviation grows with the square root of the in-degree $k/2$, $\sigma=\sqrt{k/2 \cdot q(1-q)}$, while the sample space grows linearly, $\sim k/2$. Therefore, if we look instead at a defaulted fraction of total in-neighbours $k/2$, its standard deviation behaves like $\sqrt{\frac{q(1-q)}{k/2}} \underset{k \to \infty}{\longrightarrow} 0$. Therefore, as the network degree goes to infinity, the probability distribution of the fraction of defaulted neighbours gets more and more localized around its expected value.
Thus, for a large enough $k/2$, we can approximate the CDF of a binomial function with the Heaviside function $\Theta(x)$, and we use the half maximum convention, $\Theta(0) = 1/2$:
\begin{equation}
    \left[1 - F_{CDF}^{BINOM}\left(m_{\mu},\frac{k}{2},q \right)\right] \approx \Theta\left(\frac{k}{2} \cdot q - \left\lceil \frac{k}{2} \frac{{\epsilon}_{\mu}}{{\Lambda}_b} \right \rceil\right) \approx \Theta\left(\frac{k}{2}\left(q - \frac{{\epsilon}_{\mu}}{{\Lambda}_b}\right)\right)
    \label{eqn:binom_approx}
\end{equation}

For the factor ${\Pi}_\mu$ in the Equation \ref{eqn:q_N}, setting $N \to \infty$ with similar arguments as above can be shown to reflect in ${\Pi}_\mu \to E({\Pi}_\mu) = E(N_\mu)/N = N p_\mu/N = p_\mu$.
Thus, the limiting value of the average default probability in the case of an infinite network, is simply given by Equation \ref{eqn:lim_nocorr}

In case the correlations are present, we insert the approximation from Equation \ref{eqn:binom_approx} into Equation \ref{eqn:q_corr}:
\begin{align}
    \langle q_{lim} \rangle &= \int_{-\infty}^{\infty} f_X(\alpha) \cdot \left[ {\pi}_1(\alpha) + \sum_{\mu=2} {\pi}_{\mu}(\alpha) \cdot \prod_{\phi = 2}^{\mu} \Theta\left(\sum_{\nu=1}^{\phi-1} {\pi}_{\nu}(\alpha) - \frac{{\epsilon}_{\phi}}{{\Lambda}_b}\right)\right] \,d\alpha  \\ \nonumber
    &= \int_{-\infty}^{\infty} f_X(\alpha) {\pi}_1(\alpha)\,d\alpha + \sum_{\mu=2} \int_{-\infty}^{\infty} f_X(\alpha) \cdot {\pi}_{\mu}(\alpha) \cdot \prod_{\phi = 2}^{\mu} \Theta\left(\sum_{\nu=1}^{\phi-1} {\pi}_{\nu}(\alpha) - \frac{{\epsilon}_{\phi}}{{\Lambda}_b}\right) \,d\alpha \\ \nonumber
    &= \int_{-\infty}^{\infty} f_X(\alpha) {\pi}_1(\alpha)\,d\alpha + \sum_{\mu=2} \int_{- \infty}^{l_{\mu}} f_X(\alpha) \cdot {\pi}_{\mu}(\alpha)\,d\alpha
\end{align}
The finite upper limits $l_{\mu}$ in the integrals result from the theta function. For the case of $\mu = 2$ we get from the condition:
\begin{align}
    {\pi}_1(\alpha) &> \frac{{\epsilon}_2}{{\Lambda}_b} \nonumber \\ 
    F_Y(F_Z^{-1}(p_I) - \alpha) &> \frac{{\epsilon}_2}{{\Lambda}_b} \\ \nonumber
    \alpha &< F_Z^{-1}(p_I) - F_Y^{-1}\left( \frac{{\epsilon}_2}{{\Lambda}_b} \right) \coloneqq l_{II}
\end{align}
The term for $\mu = 3$, consists of a product of two theta functions, with the first resulting in the limit $l_{II}$ and the second being the solution of ${\pi}_2(\alpha) + {\pi}_3(\alpha)>{\epsilon}_3/{\Lambda}_b$:
\begin{equation}
     \alpha < F_Z^{-1}(p_{II}) - F_Y^{-1}\left( \frac{{\epsilon}_3}{{\Lambda}_b} \right) \coloneqq l_{III}
\end{equation}
Due to the product of theta functions, the upper limit in the integral is determined by the smaller of the two values $l_{II}$ and $l_{III}$. The total expected fraction of default in the limit of an infinite network for $\mu = 1,2,3$ is represented by Equation \ref{eqn:lim_corr}
\end{proof}

To make the relation between CVNA and CVDA clearer, let us consider, for example, the case with with the following parameters: $\vec{p} = (0.02,0.09,0.89)$, $\vec{\sigma} = (-1.1,-0.75,0.0)$, ${\Lambda}_b = 8$, $\rho = 0.1$, on an infinite network. Banks' average default probability, after taking into account the network, equals to $\langle q_{lim} \rangle = 0.28$ and is thus more than 10 times larger than the external probability of default $p_1 = 0.02$, which, given Equation \ref{eqn:CVNA-CVDA}, simply translates into CVNA being more than 10 times larger than CVDA. 

\subsection{Finite probability of default is not system-size dependent}\label{s:size_effect}

\begin{figure*}[!htbp]
\begin{subfigure}{.5\textwidth}
  \centering
  \includegraphics[width=1\linewidth]{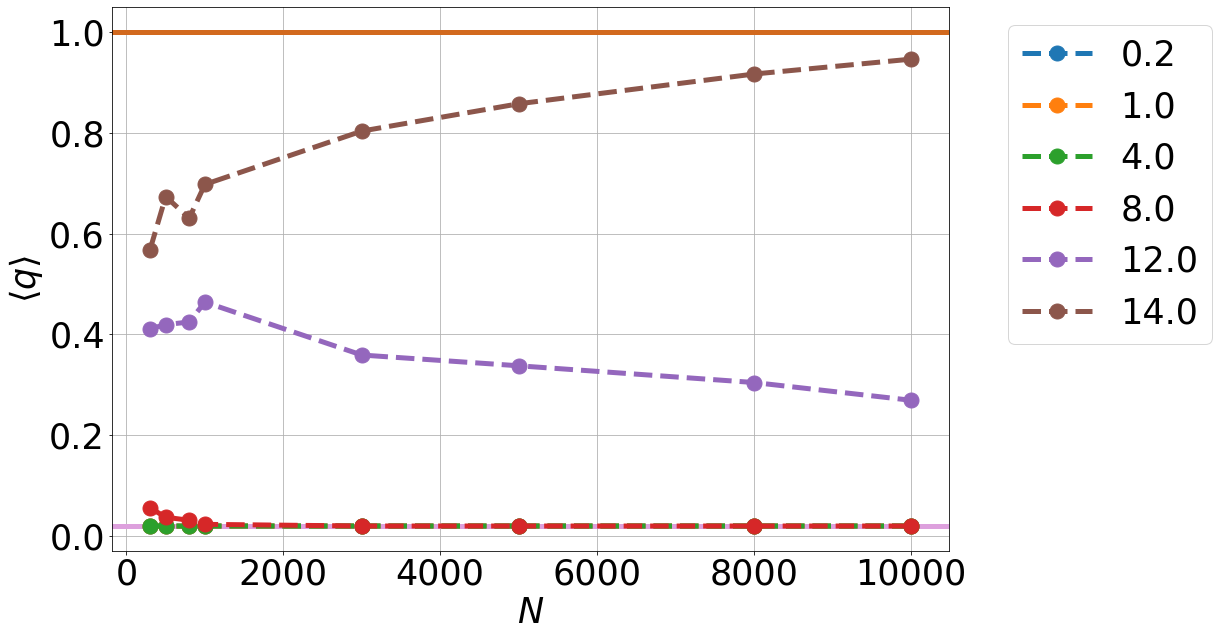}
  \caption{$\rho = 0.0$}
  \label{fig:system-first}
\end{subfigure}%
\begin{subfigure}{.5\textwidth}
  \centering
  \includegraphics[width=1\linewidth]{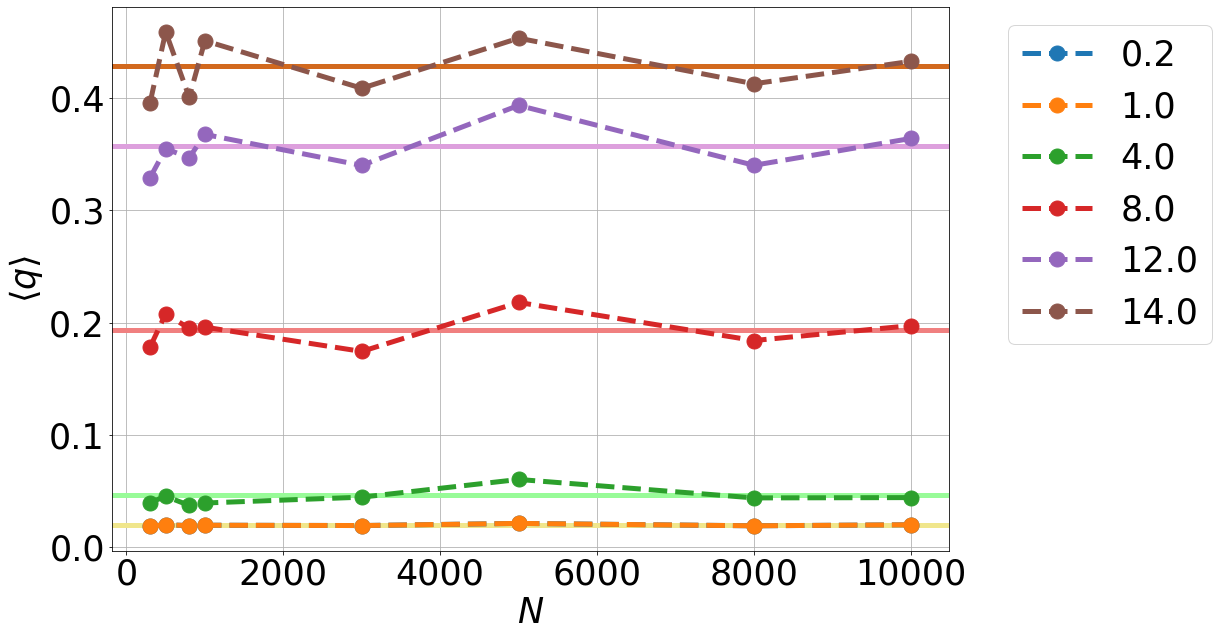}
  \caption{$\rho = 0.1$}
  \label{fig:system-second}
\end{subfigure}
\begin{subfigure}{.5\textwidth}
  \centering
  \includegraphics[width=1\linewidth]{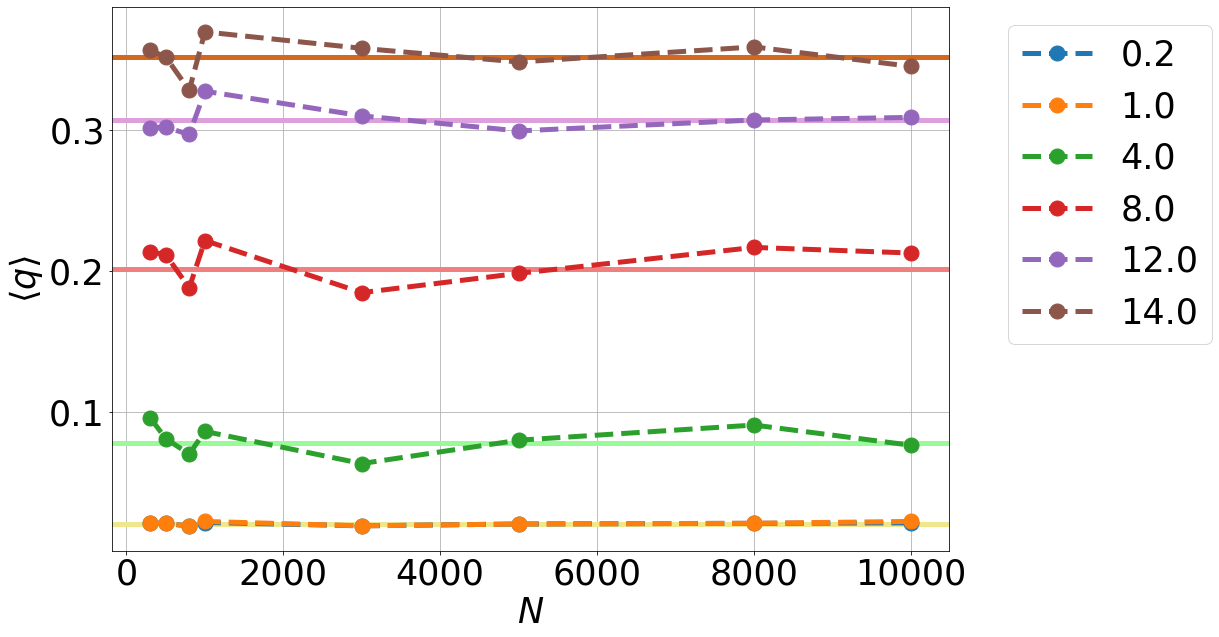} 
  \caption{$\rho = 0.2$}
  \label{fig:system-third}
\end{subfigure}%
\begin{subfigure}{.5\textwidth}
  \centering
  \includegraphics[width=1\linewidth]{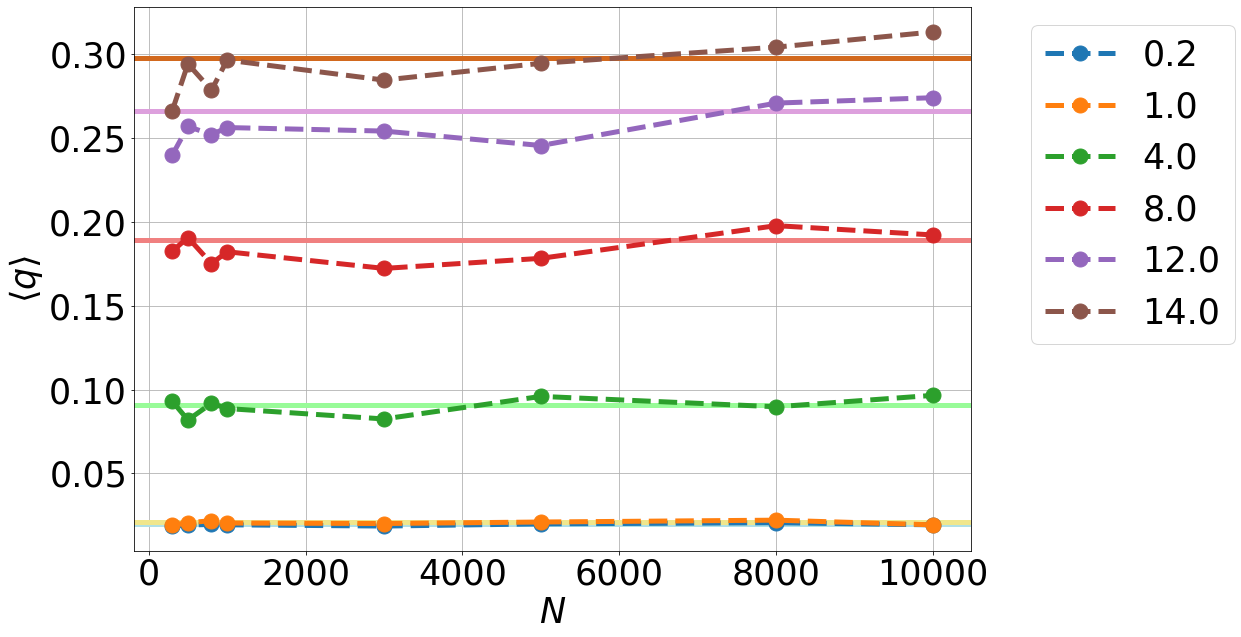} 
  \caption{$\rho = 0.3$}
  \label{fig:system-fourth}
\end{subfigure}%
\caption{\textbf{Probability of default and the system scale} Probability of default is plotted on the y-axis for maximally diversified (complete) networks with the numbers of banks $N = \{ 300,500,800,1000,3000,5000,8000,10000 \}$ on the x-axis. At time $t = 0$ the shock values are $\vec{\sigma} = (-1.1, -0.75, 0)$ , with probabilities $\vec{p} = (0.02, 0.09, 0.89)$, respectively. In different subplots (\textbf{a)}-\textbf{d)}) we vary the correlation coefficient from $0$ to $0.3$. The subplots indicate that the correlation-induced probability of default does not depend on the system size, which is further supported with horizontal lines that represent the analytical result for the limit of an infinite network.}
\label{fig:system_size}
\end{figure*}

We show how the expected fraction of default scales with the size of the system with correlation (Fig. \ref{fig:system-second}, \ref{fig:system-third},\ref{fig:system-fourth}) and without correlation (Fig. \ref{fig:system-first}). We plot expected fractions of default on complete networks for network sizes $N = \{ 300,500,800,1000,3000,5000,8000,10000 \}$. The shock values and probabilities used in simulations are $\vec{\sigma} = (-1.1, -0.75, 0)$ , and $\vec{p} = (0.02, 0.09, 0.89)$. We vary the correlation levels $\rho = \{ 0,0.1,0.2,0.3 \}$ and use the set of interbank leverages ${\Lambda}_b = \{ 0.2,1,4,8,12,14 \}$. The simulation results are plotted as points connected with dashed lines. In addition to that we depict theoretical results in the limit of an infinite network with horizontal lines. 

It is clear from Fig. \ref{fig:system-first} that the system size plays a role when the correlations are absent. The theoretical limits show two limits for convergence, $\langle q_{lim} \rangle = \{ 0.02,1.0 \}$, which correspond to no contagion and the entire network defaulting, respectively. We can see the dashed lined approaching either of the two limits as the system size increases.

The introduction of correlations changes the picture, and in  Fig. \ref{fig:system-second}, \ref{fig:system-third}, \ref{fig:system-fourth} we see that the expected fraction of default remains approximately constant for all the system sizes $N$, and that it corresponds to the theoretical limits.

\begin{figure}[!htbp]
\centering
\begin{subfigure}{\textwidth}
\includegraphics[width=1\linewidth]{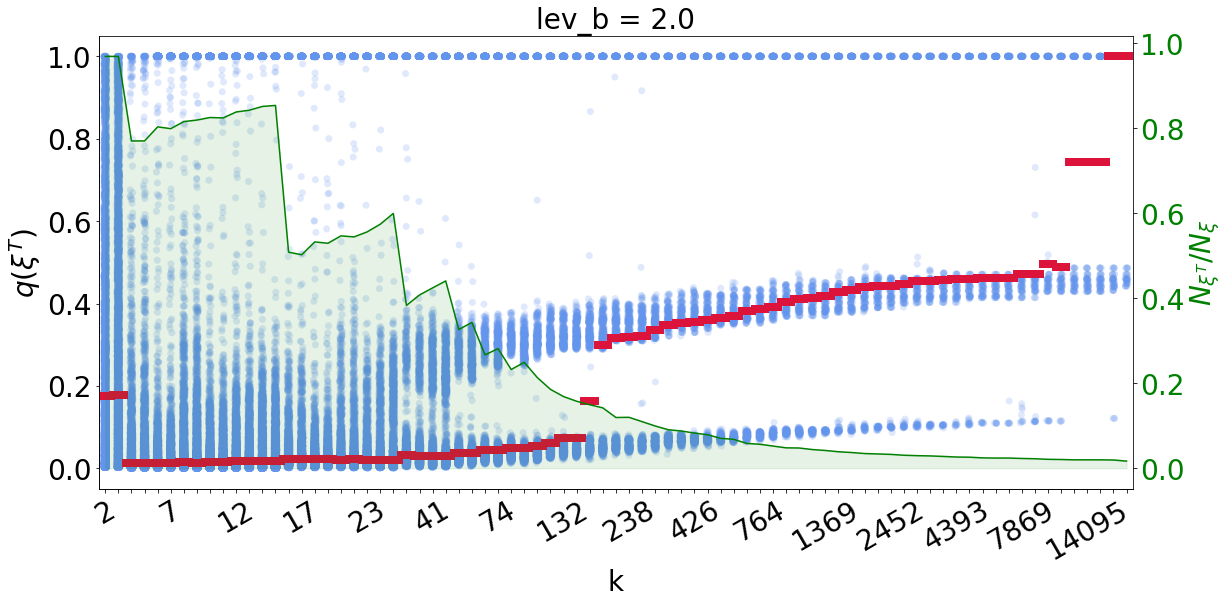}
\caption{${\Lambda}_b = 2$}
\label{fig:distribution-a}
\end{subfigure}
\begin{subfigure}{.33\textwidth}
  \centering
  \includegraphics[width=1\linewidth]{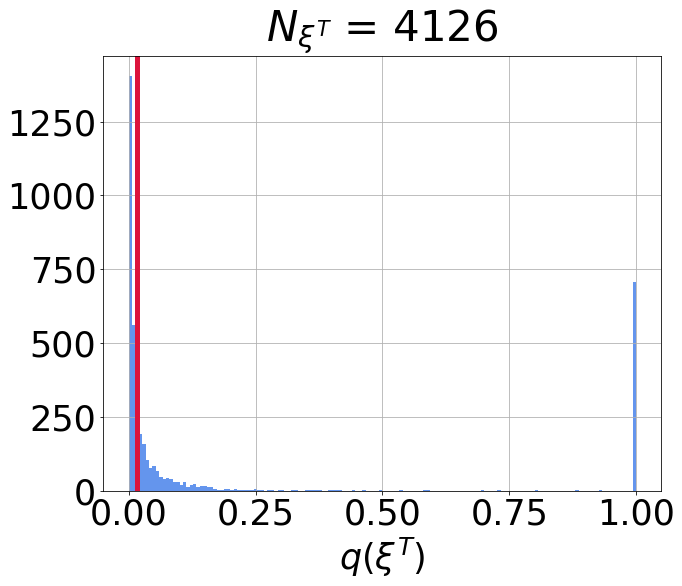} 
  \caption{$k = 10$}
  \label{fig:distribution-b}
\end{subfigure}%
\begin{subfigure}{.33\textwidth}
  \centering
  \includegraphics[width=1\linewidth]{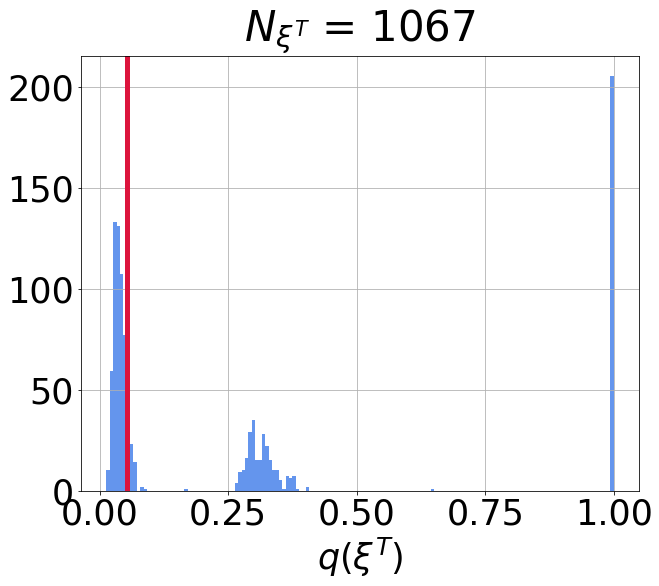} 
  \caption{$k = 105$}
  \label{fig:distribution-c}
\end{subfigure}%
\begin{subfigure}{.33\textwidth}
  \centering
  \includegraphics[width=1\linewidth]{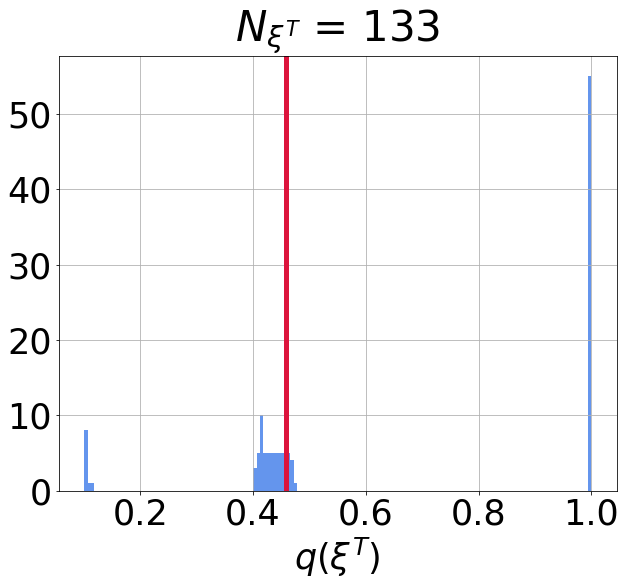} 
  \caption{$k = 4393$}
  \label{fig:distribution-d}
\end{subfigure}%
  \caption{\textbf{Analysis of the simulation results.} We present results from a default process with parameters $\vec{p} = (0.02, 0.09, 0.89), \vec{\sigma} = (-1.1, -0.75, 0.0)$, ${\Lambda}_b$, $\rho = 0.3$. In panel \textbf{a)} we present the fraction $N_{{\xi}^T}/N_{\xi}$ of the simulation realizations that contained network propagation with the green line and the shaded area. On top of that, we plot in blue the one-dimensional histograms containing individual fractions of default from realizations with propagation, for each network degree $k$. The red line represents the median value of those fractions of default. In panel \textbf{b)},\textbf{c)} and \textbf{d)}, for network degrees $k = 10, 105, 4393$ we depict the one-dimensional histograms from the panel \textbf{a)} in two dimensions, stating the total number of processes with propagation in the titles.}
\label{fig:distribution}
\end{figure}

\begin{figure}[!htbp]
    \centering
    \includegraphics[width = \textwidth]{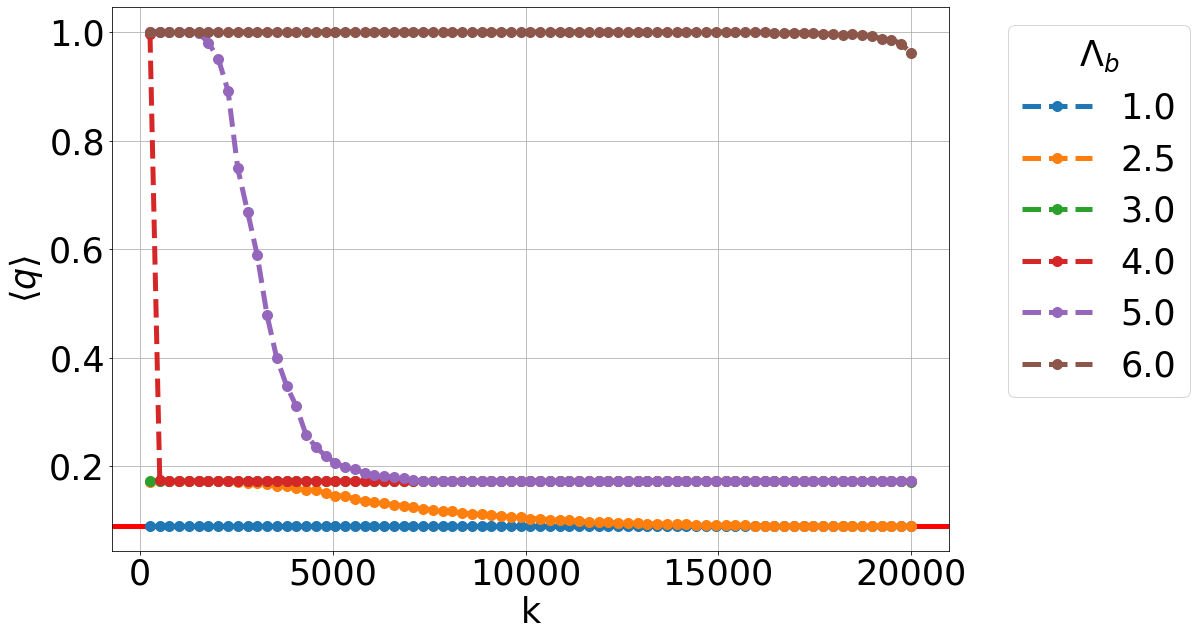}
    \caption{\textbf{Convergence limits.} We show in this figure that, for the set of parameters ($\vec{p} = (0.09, 0.083, 0.827)$, $\vec{\sigma} = (-1.1, -0.75, 0)$) and with no correlation, the system converges to three possible limits after diversification, depending on the interbank leverage ${\Lambda}_b$. On y-axis we plot the expected fraction of default, and on the x-axis the network degree, ranging from $k=252$ to the complete network $k=19998$.}
    \label{fig:third_compartment}
\end{figure}

\begin{figure}[!htbp]
  \centering
  \includegraphics[width=1\linewidth]{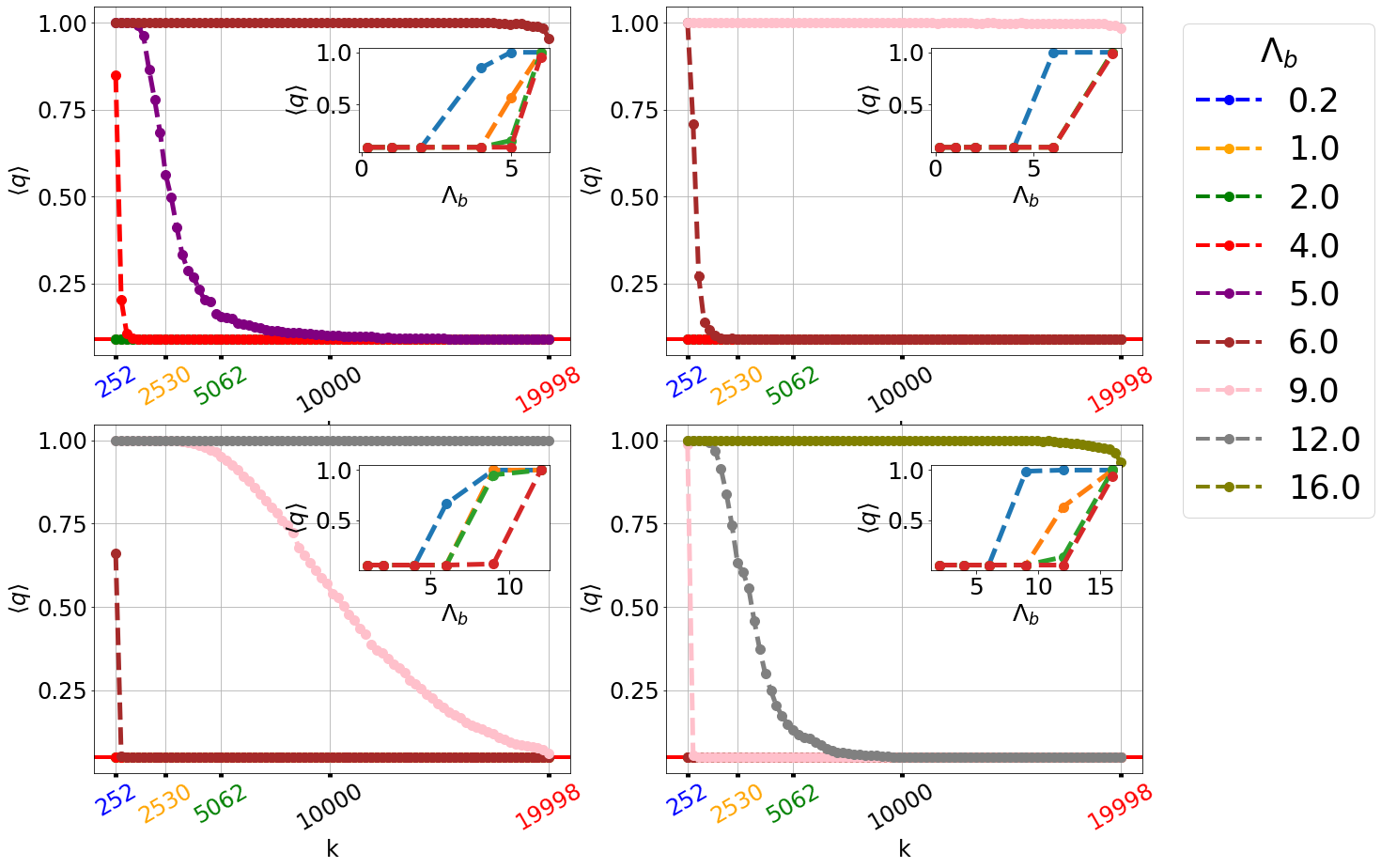}
    \caption{\textbf{Switching between sub- and supercriticality.}
    On x-axis the degree of the network starts from $252$ and goes up to a complete network, $k=19998$, and on the y-axis we plot the expected fraction of default, $\langle q \rangle$. In the top panel the initial probability of default is fixed, $\vec{p} = (0.09, 0.083, 0.827)$ and we change the shock sizes from left panel $\vec{\sigma} = (-1.1, -0.5, 0)$ to right panel $\vec{\sigma} = (-1.1, -0.25, 0)$. In the lower panel we use a different initial probability of default $\vec{p} = (0.05, 0.086, 0.864)$ but for the same shock sizes. Interbank leverage takes values from the range in ${\Lambda}_b = \{0.2,1,2,4,5,6,9,12,16\}$ and it is shown that with the decrease of the initial probability, or the shock size, the system behaviour stays the same if we increase ${\Lambda}_b$ appropriately. In the panel insets, we plot the dependence of $\langle q \rangle$ on the interbank leverage ${\Lambda}_b$, for the chosen values of $k$ that are marked on the x-axis of the plot containing the inset.} 
    \label{fig:sub_sup}
\end{figure}

\subsection{General form of results}

We simulated the process described in the Model section on a network of $N = 10000$ banks. The network size is chosen to be the largest possible, in accordance with the computational constraints, to reduce the finite system effects. The fraction of defaulted vertices is obtained as a result of each of the $N_{\xi} = 5000$ simulations performed for each network degree $k$. The final result for each $k$ is a distribution of the default fractions. 

We first present the result for $\vec{p} = (0.02, 0.09, 0.89), \vec{\sigma} = (-1.1, -0.75, 0.0)$, ${\Lambda}_b = 2$ and $\rho = 0.3$. To describe the shape of the results in Fig. \ref{fig:distribution-a}, we use only the subset of the results in which the propagation did happen on the network.
We stress that in a sample of $N_{\xi}$ processes that start on a network, network propagation of default does not happen for every one of them, and we mark with $N_{{\xi}^T}$ the number of realizations ${\xi}^T$ in which it does happen. 
With the green shaded area we mark the fraction of the process realizations in which the propagation on the network was triggered, $N_{{\xi}^T}/N_{\xi}$.
We denote the fraction of default for such process realizations ${\xi}^T$ with $q({\xi}^T)$.
The blue dots represent a one dimensional histogram of the individual default fractions resulting from propagation. The red line shows the median value of those fractions. The median is shown instead of the mean, since it is a more appropriate statistic due to the multimodal shape of the data. In the Figures \ref{fig:distribution-b},\ref{fig:distribution-c},\ref{fig:distribution-d}, we show the histograms and the median value in two dimensions for three different values of network degree.
We can conclude from this figure that, in the case with correlation, diversification lowers the probability that propagation occurs in the system. Nevertheless, the probability of default conditioned on propagation increases.

Further on, using all results obtained from the simulations, we show that, for uncorrelated external shocks, the diversified systems end up in either a subcritical or a supercritical regime, i.e. either with diversification the expected fraction of default gets reduced to some limit, or the entire system defaults every time regardless of it. The regime the system chooses depends on the process and network parameters. 

Finally, when correlation is introduced to the external shocks, we show that, in a previously well diversified system, the expected fraction of default takes up a finite value regardless of the diversification.


The process and network parameter ranges were chosen from the empirical values of the real world interbank and asset markets.

\subsection{Limiting values of convergence for the probability of defaults}

Depending on the parameters of the network (interbank leverage) and the process (shock sizes and probabilities), after the complete diversification the probability of default will converge to one of the limiting values. The maximal number of the possible limiting values of default probability is equal to the number of different shock values that we use in our model. For example, for the shock $\vec{p} = ( p_1, p_2, p_3 )$, the limiting values are equal to $p_1$, $p_1+p_2$ and $p_1 + p_2 + p_3 =1$. The interpretation of this phenomenon lies in the fact that the initial shock with 3 possible values splits the $N$ banks into compartments with 3 different levels of vulnerability $(N_1,N_2,N_3)$. The compartment $N_1$ always defaults because the shock $p_1$ causes default right away. 
In addition to that, the compartment $N_2$ has the equity decreased to ${\epsilon}_2$, and will default if a fraction ${\epsilon}_2/{\Lambda}_b$ of its interbank assets is lost. If the probability of initial default $p_1$ is equal or larger than the fraction ${\epsilon}_2/{\Lambda}_b$ (in a large and complete network), every bank in $N_2$ will have a fraction of initially defaulted borrowers (i.e. $p_1$) large enough to drive it into default. 
Therefore, for some combinations of the shock probabilities and sizes, and the interbank leverage, beside the default of the compartment $N_1$ we can have either the compartment $N_2$, or both $N_2$ and $N_3$ (the whole system) going into default as shown in Fig. \ref{fig:third_compartment}

\subsection{Subcritical and supercritical regime result from uncorrelated shocks }

Depending on the parameters of the network (interbank leverage) and the process, the probability of default can either decrease to some limit with diversification, or remain at the highest level, equal to one. We refer to these two modes of behaviour as subcritical and supercritical regimes. For a fixed probability and size of the shocks, increasing the interbank leverage leads from the subcritical into the supercritical mode. The same effect, switching from a subcritical into supercritical regime, is obtained with increasing the probability of initial default $p_0$, and increasing the size of the shock ${\epsilon}_1$.

In Fig. \ref{fig:sub_sup} we present this effect on four panels by varying the initial probability in the vertical direction ($\vec{p} = (0.09, 0.083, 0.827)$ in the top panel to $\vec{p} = (0.05, 0.086, 0.864)$ in the bottom panel) and the shock size in the horizontal direction ($\vec{\sigma} = (-1.1, -0.5, 0)$ in the left panel to $\vec{\sigma} = (-1.1, -0.25, 0)$ in the right panel). We see that an appropriate choice of the interbank leverage ${\Lambda}_b$ recovers both regimes regardless of the intial probability of the shocks and their sizes.

Therefore, we see that maximal diversification can mitigate the effect of the CVA network component only if the interbank leverage is in the subcritical range for the given parameters of the shock (probability and size).

\subsection{Correlated shocks cause a non-vanishing probability of default}

\begin{figure*}[!htbp]
\begin{subfigure}{.5\textwidth}
  \centering
  \includegraphics[width=1\linewidth]{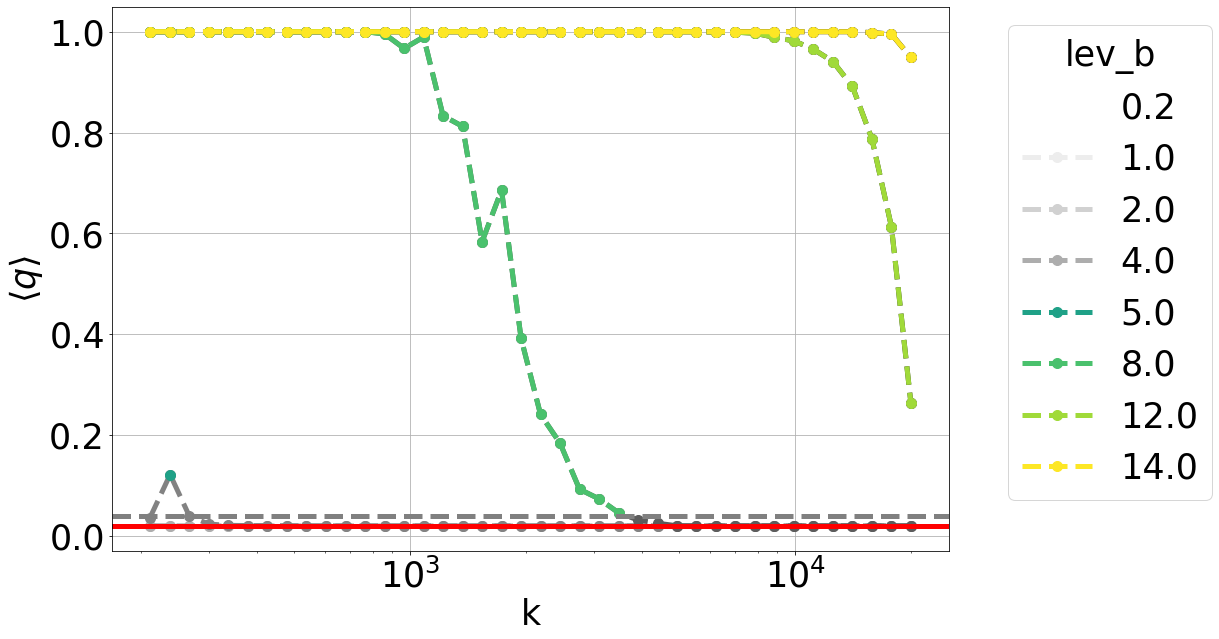}
  \caption{$\rho = 0$}
  \label{fig:corr-first}
\end{subfigure}%
\begin{subfigure}{.5\textwidth}
  \centering
  \includegraphics[width=1\linewidth]{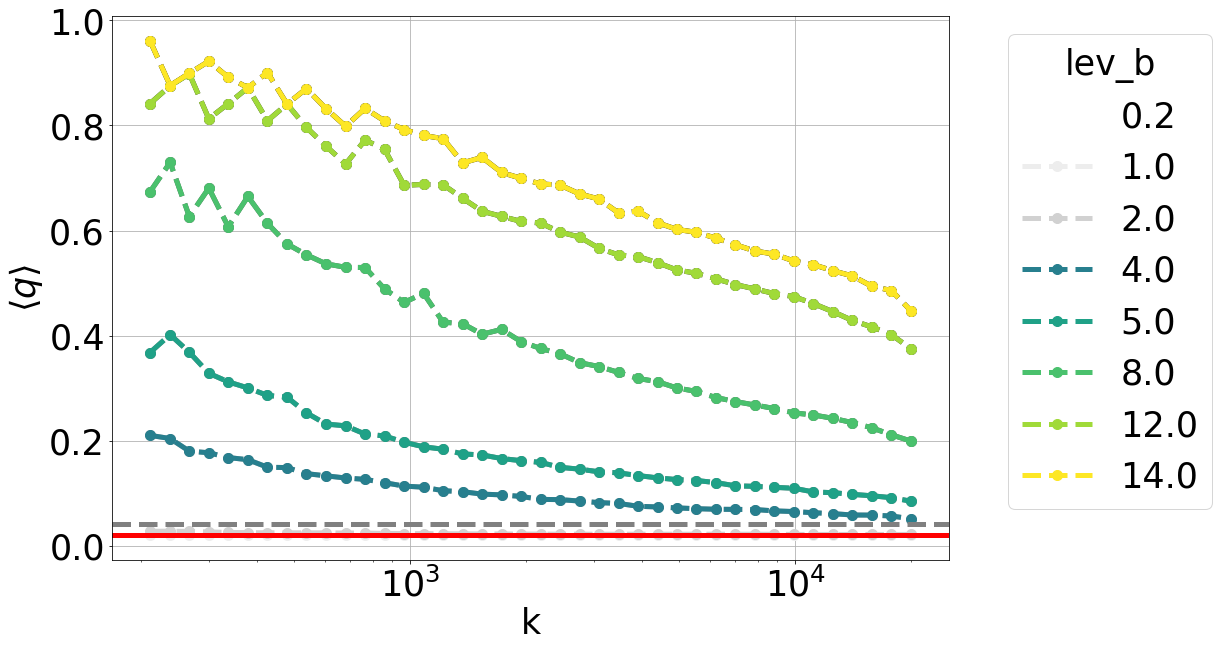} 
  \caption{$\rho = 0.1$}
  \label{fig:corr-second}
\end{subfigure}
\begin{subfigure}{.5\textwidth}
  \centering
  \includegraphics[width=1\linewidth]{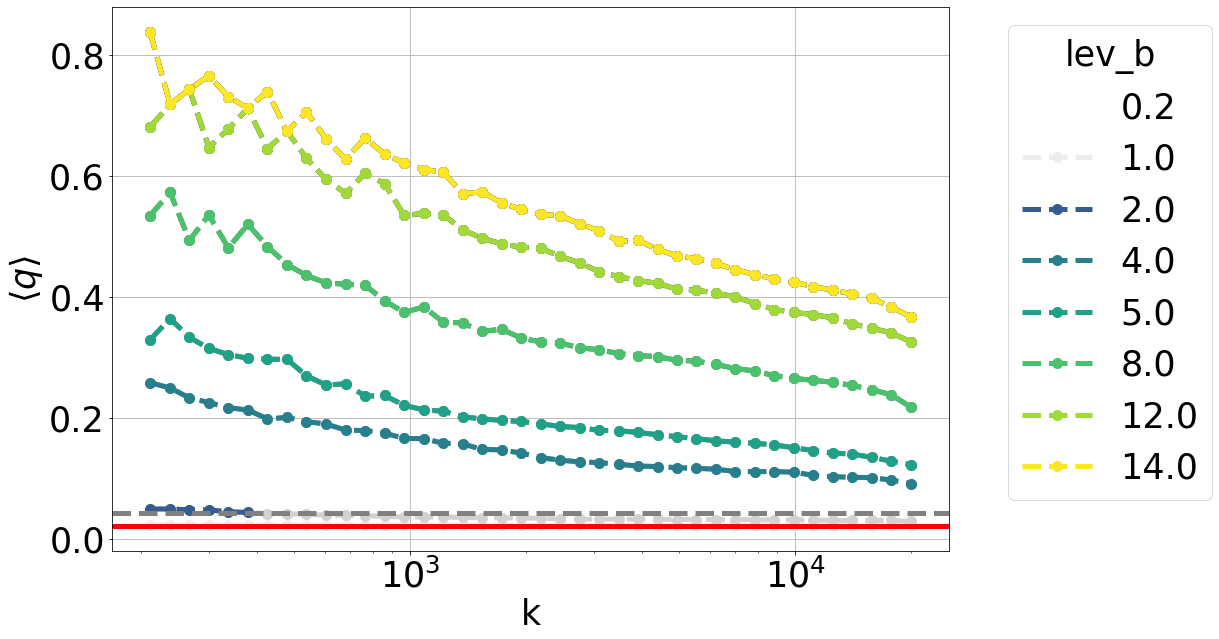} 
  \caption{$\rho = 0.2$}
  \label{fig:corr-third}
\end{subfigure}%
\begin{subfigure}{.5\textwidth}
  \centering
  \includegraphics[width=1\linewidth]{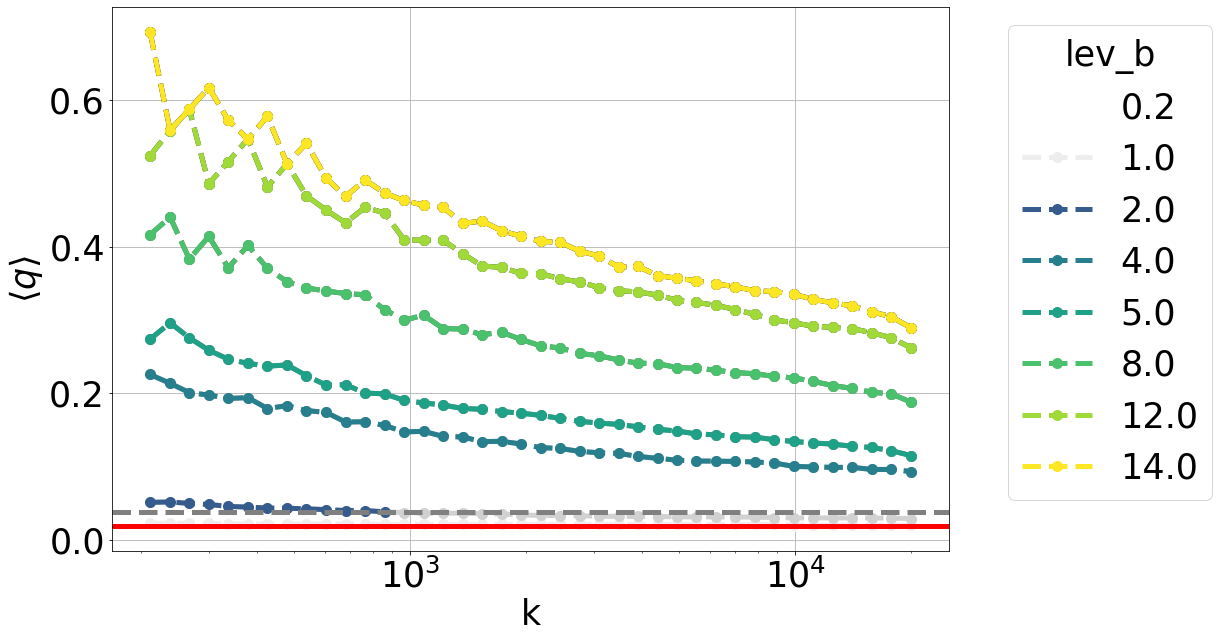} 
  \caption{$\rho = 0.3$}
  \label{fig:corr-fourth}
\end{subfigure}%
\caption{\textbf{$N = 10000$} Probability of defaults with respect to the correlation between shocks. The degree of the vertices $\langle k \rangle$ starts from $211$ and goes up to $19998$ (complete network), for the  number of vertices $N = 10000$. The plots do not show a smaller $\langle k \rangle$ than $211$, as the transitional effects in that range are due to the choice of the network, and not relevant to real systems.  The shock values are $\vec{\sigma} = (-1.1, -0.75, 0)$ , with probabilities $\vec{p} = (0.02, 0.09, 0.89)$, respectively. We vary the interbank leverage ${\Lambda}_b = \{0.2,1,2,4,5,8,12,14\}$. The red horizontal line represents the initial probability of default (without taking into account the network of liabilities). The grey dashed line represents double the values of the initial probability of default. Subfigures \textbf{a)-d)} show different correlation coefficients $\rho = \{0, 0.1, 0.2, 0.3\}$. Values where the introduced correlation increases the probability of default by $100 \%$ or more are shown in colour, values below that are shown in greyscale.}
\label{fig:corr_varying_double}
\end{figure*}
We model the correlation on the equity levels using the Gaussian copula, as described in the Model section.

In a real market, banks invest in multiple assets outside the banking network. Their portfolios can overlap, and the values of the assets, and thus their default probabilities, can be correlated. However, simulating such a system for realistic parameters can be computationally too expensive.

Nevertheless, it is clear that a correlation between the shocks on external assets would lead to the equity levels of banks being correlated after the shock. We show in Appendix Fig. \ref{fig:rho} that the empirically measured sample correlation between the equity values ${\rho}^E$ is a non-decreasing function of the empirically measured sample correlation ${\rho}^A$ imposed on the external assets. We argue that for every joint probability function of the equity levels, we can find multiple joint probability functions of the external assets, and that justifies reducing the model for correlations of the shocks to equity level.

For the shock levels that we use in simulations, we choose parameters realistic for markets,  values $\vec{\sigma} = (-1.1, -0.75, 0)$, with the probabilities  $\vec{p} = (0.02, 0.09, 0.89)$ and we simulate these shocks with 4 different levels of correlation $\rho = \{0,0.1,0.2,0.3\}$.

The results of the simulations are presented in Figure \ref{fig:corr_varying_double}. We plot the expected fraction of default in dependence of the network degree. The horizontal red full line represents the fraction of banks defaulted by the initial external shocks. The horizontal grey dashed line represents the level of probability that is double than the fraction of banks defaulted from initial shocks. 
The lines with dots, showing the expected fraction of default for different interbank leverages, are presented in two color schemes, depending on whether their value exceeds double the initial expected fraction (the grey dashed line) or not.
In the area where default fraction coming from the network is smaller than double the initial default fraction they are presented in greyscale, and in the area where it is larger than double, they are in colour.
Figures \ref{fig:corr-second}, \ref{fig:corr-third} and \ref{fig:corr-fourth} show how the introduction of correlation on the equity values raises the level of the expected fraction of network default for leverage values that were previously well diversified, lines ${\Lambda}_b = \{0.2,1,2,4,5\}$, on Fig. \ref{fig:corr-first}. 

On the other hand, if the interbank leverage ${\Lambda}_b$ was high enough to maximize the expected fraction of default without the correlations (${\Lambda}_b = 14$ on Fig. \ref{fig:corr-first}), the introduction of correlations will reduce it. 
In the case with no correlation, the fraction of initially defaulted banks comes from the binomial distribution, and for the large systems deviates very little from the expected value, which is in this case equal to $p_0$. If the expected fraction of initially defaulted banks is enough to default the entire system, and, as we concluded, the initial fraction of defaults does not vary much between realizations, the system will default in every realization and the expected total fraction of defaults will be equal to 1.
However, if the probabilities of initial defaults of banks are correlated, although the expected value of the fraction of defaults stays the same, the fractions of default are no longer drawn from a binomial, but from a distribution in the Equation \ref{eqn:p_corr}, which has positive skewness. Therefore, realizations where the fractions of initially defaulted banks are very close to zero and thus not large enough to trigger the default of the entire system become very likely. Taking an expectation over the sample of all realizations necessarily gives a value less than 1.

\section{Conclusion}

In this paper, we have studied the behavior of Credit Valuation Network Adjustment (CVNA) as an extension of CVA to context of a network of financial contracts, also in the presence and correlated and non-correlated shocks. We have shown that for the broad range of parameters, the CVNA is significantly larger than CVA, and should be taken into account. We have also demonstrated that the results are not finite size effects, using the mapping of the studied dynamics to the mean-field approximation of the threshold model of \citet{watts2002threshold}. Although this model was not solvable analytically, it presents an interesting alternative for the numerical evaluation of CVNA in financial networks.

We have then shown that, depending on the leverage, the presence of correlations can either increase or reduce the probabilities of defaults compared to the case of uncorrelated shocks.


We believe that the effects of correlated shocks on CVNA and systemic risk demonstrated in this paper are important enough for practitioners to warrant future research in this direction. An important future contribution would consider financial networks with heterogeneous agents drawn from real data and the way correlations in shocks affect  CVNA of the contract. Even further, applying CVNA in the case of a dynamic financial system for new contracts could in principle affect the evolution of the financial system network structure, which we believe is an important academic question. 

\section{Acknowledgments}
VZ wish to acknowledge the support of the Croatian Science Foundation (HrZZ) Projects No. IP-2019-
4-3321. VZ and IB acknowledge support form QuantiXLie Centre of Excellence, a Project co-financed by the Croatian Government and European Union through the European Regional Development Fund - the Competitiveness and Cohesion Operational Program (Grant KK.01.1.1.01.0004, element leader N.P.). IB acknowledges support by the Federal Commission for Scholarships for Foreign Students (FCS) through the Swiss Government Excellence Scholarship.

\bibliographystyle{apalike}   
\bibliography{references}

\section{Appendix}
\subsection{Appendix A}

\begin{figure*}[h!]
\centering
\includegraphics[width=0.5\linewidth]{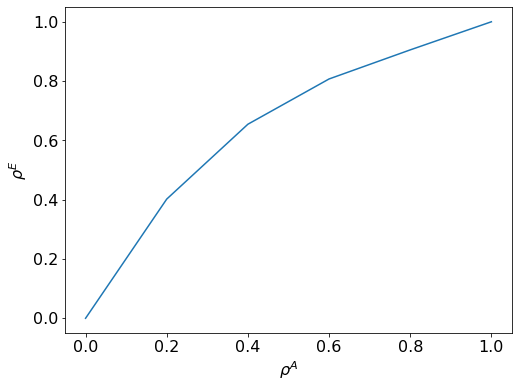}
\caption{Dependence of the correlation on the equity values ${\rho}^E$ on the correlation on the external assets ${\rho}^A$}
\label{fig:rho}
 \end{figure*}

\subsection{Appendix B}

For the case $\rho = 0$ (independent shocks), when there is no correlation on the shock levels, we can show that the distribution of the counts of banks $N_{\mu}$ hit with each shock ${\sigma}_{\mu}$ equals the multinomial probability distribution.
 
If we set the probabilities of the shock values $\vec{\sigma} = ({\sigma}_1,{\sigma}_2,{\sigma}_3)$, that result in equity values $\vec{\epsilon} = 1 -\vec{\sigma}$, to be $\vec{p} = (p_1,p_2,p_3)$, the cumulative distribution function of the random variable of the shock $\mathcal{S}$ will be, as in \ref{eqn:shock_sample}:
\begin{equation}
    F_{\mathcal{S}}(x) = 
\begin{cases} 
        p_1 \coloneqq p_I, & x = {\sigma}_1 \\
        p_1 + p_2 \coloneqq p_{II}, & x = {\sigma}_2 \\
        1, & x = {\sigma}_3.
   \end{cases}
\end{equation}

We use the inverse transform sampling to sample the shock random variable $\mathcal{S}$, i.e. we transform the uniform random variable $U$ using the inverse CDF $F_{\mathcal{S}}^{-1}(U) = \mathcal{S}$.

For $N$ banks, we write the total probability that $N_1$ banks end up being shocked by ${\sigma}_1$ with probability $p_1$, $N_2$ banks get shocked by ${\sigma}_2$ with probability $p_2$ and $N_3$ banks get shocked by ${\sigma}_3$ with probability $p_3$ as:

\begin{align}
    \mathcal{P}&(\vec{N};N;\vec{p}) = \nonumber \\
    &= \binom{N}{N_1,N_2,N_3} P(U_1 \leq p_I, \dotsc, p_I < U_{i} \leq p_{II},\dotsc, U_N > p_{II})\nonumber \\
    &= \binom{N}{N_1,N_2,N_3}P(U_1 \leq p_I, \dotsc, p_I < U_{i} \leq p_{II}, \dotsc,  U_N \leq 1 - p_{II})\nonumber\\
    &= \binom{N}{N_1,N_2,N_3} \left[C(p_I,\dotsc, p_{II}, \dotsc, 1-p_{II}) - C(p_I, \dotsc, p_I,\dotsc, 1-p_{II}) \right]\nonumber\\
    &= \binom{N}{N_1,N_2,N_3} \left[p_I \dotsm   p_{II}  \dotsm (1-p_{II}) - p_I \dotsm p_I \dotsm  (1-p_{II})\right]\nonumber\\
    &= \binom{N}{N_1,N_2,N_3} p_I \dotsm  (p_{II}-p_I) \dotsm  (1-p_{II}) \nonumber\\
    &= \binom{N}{N_1,N_2,N_3} \underbrace{p_1 \dotsm p_1}_{N_1} \cdot \underbrace{p_2 \dotsm p_2}_{N_2} \cdot \underbrace{p_3 \dotsm p_3}_{N_3}\nonumber \\
    &= \binom{N}{N_1,N_2,N_3}\quad p_1^{N_1} \cdot p_2^{N_2} \cdot p_3^{N_3}.
    \label{eqn:p_noncorr}
\end{align}
We expressed the probability using the uniform variables that we use for sampling. The multinomial coefficient counts the number of ways we can get the same counts, $N_1$, $N_2$ and $N_3$, from a set of $N$ banks. In the third row we utilize the definition of the copula, and switch to the fourth row using the fact that the independence of all the variables in the copula leads to the explicit form of the copula to be simply a product of all those variables \citep{schmidt2007copula}. From there on, with a bit of algebra, and 
by returning to the originally defined probabilities $p_1,p_2,p_3$, we recover a multinomial probability mass function.

\subsection{Appendix C}

In the Appendix B, we obtained the distribution of the numbers of banks $\vec{N} = (N_1,N_2,N_3)$ affected by the set of shocks $\vec{\sigma} = ({\sigma}_1,{\sigma}_2,{\sigma}_3)$ that occurred with probabilities $\vec{p} = (p_1,p_2,p_3)$, in the case when all the shocks were independent. Now we will look more closely into the case when correlation between the shocks is introduced. We wish to obtain the probability mass function $\mathcal{P}(N_1,N_2,N_3;N;p_1,p_2,p_3;\rho)$ for the numbers of banks $\vec{N} = (N_1,N_2,N_3)$ in compartments $\mu=1,2,3$ when correlation $\rho$ is present.

In the model, we introduced the correlations between shocks that reduce equity values by employing the Gaussian copula $C^{Ga}_{\Sigma}(\textbf{u})$ to correlate our uniform variables $U_i$ before we used them for sampling shocks with Equation (\ref{eqn:shock_sample}). However, the approach used to derive the distribution in the uncorrelated case (\ref{eqn:p_noncorr}) relies on the independence of the uniform variables, so we need to take a different perspective.
Therefore, instead of using the uniform random variables $U_i$ as in (\ref{eqn:p_noncorr}), we base this derivation on Gaussian random variables, $Z_i \sim \mathcal{N}(0,1)$, and use the property of a single factor copula that enables the decomposition of the correlated variables $Z_i$ into the sum of independent Gaussian variables $X \sim \mathcal{N}(0,\rho)$ and $Y_i \sim \mathcal{N}(0,1-\rho)$. The random variable $X$ is the part common to all the sampled variables, while $Y_i$ is the idiosyncratic part, specific for each bank $i$.
\begin{equation}
    Z_i = X + Y_i, \quad Z_i \sim \mathcal{N}(0,1), \quad  X \sim \mathcal{N}(0,\rho), \quad Y_i \sim \mathcal{N}(0,1-\rho).
\end{equation}
Next, we need to transform the subintervals for sampling  with limits $p_I, p_{II}$, from the $[0,1]$ segment, into the real line $\langle -\infty,\infty \rangle$ with limits $z_I,z_{II}$. Respective limits need to divide both distributions into the same percentiles:
\begin{align}
    p_1 &= p_I = F_Z(z_I), \\ \nonumber
    p_2 &= p_{II} - p_I = F_Z(z_{II}) - F_Z(z_I), \\ \nonumber
    p_3 &= 1 - p_{II} = 1 - F_Z(z_{II}).
\end{align}
From this we can get $z_I,z_{II}$ easily:
\begin{align}
    z_I &= F_Z^{-1}(p_I), \\ \nonumber
    z_{II} &= F_Z^{-1}(p_{II}).
\end{align}
We start the derivation of the probability mass function in the same way as in (\ref{eqn:p_noncorr}), by counting all the combinations of shocks on banks that can occur for the same counts $N_1$, $N_2$ and $N_3$:
\begin{align}
    \mathcal{P}&(\vec{N};N;\vec{p}) = \\ \nonumber
    &= \binom{N}{N_1, N_2, N_3} P(\underbrace{Z_1 \leq z_I, \dotsc}_{N_1}, \underbrace{z_I < Z_{i} \leq z_{II},\dotsc,}_{N_2} \underbrace{Z_{i+k+1} > z_{II}, \dotsc, Z_N > z_{II}}_{N_3}).
\end{align}
Then we look separately into the joint probability and utilize the decomposition of the random variables we defined earlier, $Z_i = X + Y_i$, so that we can express probability distribution function variables ${\gamma}_i$ as a sum ${\gamma}_i = \alpha + {\beta}_i$. The corresponding probability distribution functions for the new variables are the univariate PDF $f_X$ and the multivariate PDF $f_{\textbf{Y}}$. We first express the probability that Gaussian variables will take the values that correspond to bank counts $N_1$, $N_2$ and $N_3$ as the multidimensional integral of the the Gaussian multivariate PDF $f_{\textbf{Z}}$, and use the abbreviation $\,d\vec{\gamma} = \,d{\gamma}_1 \dotsm \,d{\gamma}_i \dotsm \,d{\gamma}_N$:
 \begin{align}
     P&(Z_1 \leq z_I, \dotsc, z_I < Z_{i} \leq z_{II},\dotsc, Z_N > z_{II})  =\\ \nonumber
     & = \int_{-\infty}^{z_I} \dotsm \int_{z_I}^{z_{II}} \dotsm \int_{z_{II}}^{\infty} f_{\textbf{Z}}({\gamma}_1,\dotsc, {\gamma}_i,\dotsc,{\gamma}_N) \,d\vec{\gamma} \\ \nonumber
     & =  \int_{-\infty}^{z_I} \dotsm \int_{z_I}^{z_{II}} \dotsm \int_{z_{II}}^{\infty}  f_{\textbf{Z}}(\alpha + {\beta}_1,\dotsc, \alpha + {\beta}_{i}, \dotsc,\alpha + {\beta}_N)  \,d\vec{\gamma} \\ \nonumber
      & =  \int_{-\infty}^{z_I} \dotsm \int_{z_I}^{z_{II}} \dotsm \int_{z_{II}}^{\infty}  \left( \int_{-\infty}^{\infty} f_{X}(\alpha) f_{\textbf{Y}}({\gamma}_1 - \alpha\dotsc, {\gamma}_i - \alpha,\dotsc,{\gamma}_N - \alpha) \,d\alpha \right)  \,d\vec{\gamma} \\ \nonumber
       & =  \int_{-\infty}^{z_I} \dotsm \int_{z_I}^{z_{II}} \dotsm \int_{z_{II}}^{\infty}  \left( \int_{-\infty}^{\infty} f_{X}(\alpha) f_{Y}({\gamma}_1 - \alpha)\dotsm f_{Y}({\gamma}_i - \alpha)\dotsm f_{Y}({\gamma}_N - \alpha) \,d\alpha \right)  \,d\vec{\gamma} \\ \nonumber
       & = \int_{-\infty}^{\infty} f_{X}(\alpha) \left( \int_{-\infty}^{z_I}  f_{Y}({\gamma}_1 - \alpha) \,d{\gamma}_1 \right) \dotsm  \left(\int_{z_I}^{z_{II}} f_{Y}({\gamma}_i - \alpha) \,d{\gamma}_{i}\right) \dotsm \left( \int_{z_{II}}^{\infty} f_{Y}({\gamma}_N - \alpha) \,d{\gamma}_N \right)  \,d\alpha  \\ \nonumber
      & = \int_{-\infty}^{\infty} f_{X}(\alpha) \left( \int_{-\infty}^{z_I}  f_{Y}({\gamma}_1 - \alpha) \,d{\gamma}_1 \right) \dotsm  \left(\int_{-\infty}^{z_{II}} f_{Y}({\gamma}_i - \alpha) \,d{\gamma}_i - \int_{-\infty}^{z_I} f_{Y}({\gamma}_i - \alpha) \,d{\gamma}_i\right) \dotsm \\ \nonumber
      & \dotsm \left( 1-\int_{-\infty}^{z_{II}} f_{Y}({\gamma}_N - \alpha) \,d{\gamma}_N \right)  \,d\alpha  \\ \nonumber
      & = \int_{-\infty}^{\infty} f_{X}(\alpha)  F_{Y}(z_I - \alpha) \dotsm  \left(F_{Y}(z_{II} - \alpha)  -  F_{Y}(z_I - \alpha) \right) \dotsm \left( 1-F_{Y}(z_{II} - \alpha) \right)  \,d\alpha  \\ \nonumber
      & = \int_{-\infty}^{\infty} f_{X}(\alpha)  (F_{Y}(z_I - \alpha))^{N_1} \cdot  (F_{Y}(z_{II} - \alpha)  -  F_{Y}(z_I - \alpha))^{N_2} \cdot ( 1-F_{Y}(z_{II} - \alpha))^{N_3}  \,d\alpha.
 \end{align}
In the third row, we replace the ${\gamma}_i$ variables with the sum $\alpha + {\beta}_i$, and proceed to express $f_{\textbf{Z}}$ as a convolution of $f_X$ and $f_{\textbf{Y}}$.
In the fifth row we use the independence of the idiosyncratic uniform random variables $Y_i$ to separate the $f_{\textbf{Y}}$ PDF into univariate $f_{Y_i} = f_Y$ PDFs. After integrating all the integrals depending on ${\gamma}_i$ in the seventh row, we get a product of $N$ CDFs of the random variable $Y$ that depend on the limits $z_I$ and $z_{II}$ according to the class $\mu$ they belong to ($N_{\mu}$ of them in each class $\mu$). The final result is an integral, going over the real axis, of the product of a PDF of the random variable $X$
and each CDF raised to its respective power $N_{\mu}$.
To make sense of the expression, we define new shock probabilities that depend on the variable $\alpha$:
 \begin{align}
 & {\pi}_1(\alpha) = F_{Y}(z_I - \alpha) = F_{Y}(F_Z^{-1}(p_I) - \alpha), \\ \nonumber
    & {\pi}_2(\alpha) = F_{Y}(z_{II} - \alpha) - F_{Y}(z_I - \alpha) = F_{Y}(F_Z^{-1}(p_{II}) - \alpha) - F_{Y}(F_Z^{-1}(p_I) - \alpha), \\ \nonumber
    & {\pi}_3(\alpha) = 1 - F_{Y}(z_{II} - \alpha) = 1 - F_{Y}(F_Z^{-1}(p_{II}) - \alpha). 
\end{align}
Finally, the complete expression for the probability mass function is structured as an integral of the multinomial distribution with varying probabilities weighted by the Gaussian distribution of the common random variable $X$:
 \begin{equation}
    \mathcal{P}(\vec{N};N;\vec{p}) = \int_{-\infty}^{\infty} f_{X}(\alpha) \binom{N}{N_1, N_2, N_3}  {\pi}_1(\alpha)^{N_1} \cdot  {\pi}_2(\alpha)^{N_2} \cdot {\pi}_3(\alpha)^{N_3}  \,d\alpha.
\label{eqn:p_corr}
\end{equation}

\end{document}